\documentclass[11pt]{article}
\usepackage{fullpage}

\usepackage{libertine}

\usepackage{amsmath,amssymb,amsthm}
\usepackage{colortbl}
\usepackage{graphicx}
\usepackage{bm}

\usepackage{enumitem}

\usepackage[breaklinks]{hyperref}
\usepackage[svgnames]{xcolor}
\usepackage[capitalise,nameinlink]{cleveref}
\hypersetup{colorlinks={true},linkcolor={DarkBlue},citecolor=[named]{DarkGreen}}
\usepackage{natbib}

\usepackage{mdframed}

 \mdfsetup{
	linewidth=1pt
}

\newtheorem{theorem}{Theorem}

\newtheorem{corollary}{Corollary}

\theoremstyle{definition}
\newtheorem{definition}{Definition}

\usepackage{multirow}
\usepackage{tikz} 

\usepackage{xcolor}
\usepackage{hyperref}
\usepackage{comment,color}

\usepackage[ruled,vlined,linesnumbered]{algorithm2e}

\makeatletter
\newcommand*{\algotitle}[2]{%
	\stepcounter{algocf}%
	\hypertarget{algocf.title.\theHalgocf}{}%
	\NR@gettitle{#1}%
	\label{#2}%
	\addtocounter{algocf}{-1}%
}
\makeatother

\usepackage{algpseudocode}
\usepackage{authblk}
\usepackage[capitalise]{cleveref}
\crefname{algorithm}{Mechanism}{mech}

\usepackage{multirow}

\allowdisplaybreaks

\makeatletter
\g@addto@macro\bfseries{\boldmath}
\makeatother

\newcount\Comments
\definecolor{darkgreen}{rgb}{0,0.7,0}
\newcommand{\kibitz}[2]{\ifnum\Comments=1\textcolor{#1}{#2}\fi}

\newcommand{\RR}{\mathbb{R}}
\newcommand{\SW}{\text{SW}}
\newcommand{\vv}{\mathbf{v}}
\newcommand{\calV}{\mathcal{V}}
\newcommand{\calA}{\mathcal{A}}
\newcommand{\dist}{\text{dist}}

\newcommand{\gaga}{\text{\sc ThresholdStepFunction}}

\title{\bf A Few Queries Go a Long Way: \\Information-Distortion Tradeoffs in Matching\thanks{\, G.~Amanatidis has been partially supported by the NWO Veni project VI.Veni.192.153. G.~Birmpas has been supported by the ERC Advanced Grant 788893 AMDROMA ``Algorithmic and Mechanism Design Research in Online Markets'', and the MIUR PRIN project ALGADIMAR ``Algorithms, Games, and Digital Markets''. 
}}

\author[1,2]{Georgios Amanatidis}
\author[3]{Georgios Birmpas}
\author[4]{Aris Filos-Ratsikas}
\author[1]{Alexandros A. Voudouris}

\affil[1]{University of Essex, United Kingdom}
\affil[2]{University of Amsterdam, The Netherlands}
\affil[3]{Sapienza University of Rome, Italy}
\affil[4]{University of Liverpool, United Kingdom}

\date{}

\begin{document}
\maketitle

\begin{abstract}
We consider the \emph{one-sided matching} problem, where $n$ agents have preferences over $n$ items, and these preferences are induced by underlying cardinal valuation functions. The goal is to match every agent to a single item so as to maximize the \emph{social welfare}. Most of the related literature, however, assumes that the values of the agents are not a priori known, and only access to the ordinal preferences of the agents over the items is provided. Consequently, this incomplete information leads to loss of efficiency, which is measured by the notion of \emph{distortion}. In this paper, we further assume that the agents can answer a small number of \emph{queries}, allowing us partial access to their values. We study the interplay between elicited cardinal information (measured by the number of queries per agent) and distortion for one-sided matching, as well as a wide range of well-studied related problems. Qualitatively, our results show that with a limited number of queries, it is possible to obtain significant improvements over the classic setting, where only access to ordinal information is given. 
\end{abstract}

\section{Introduction}\label{sec:intro}
In the \emph{one-sided matching} problem (often referred to as the \emph{house allocation} problem),  $n$ agents have \emph{preferences} over a set of $n$ items, and the goal is to find an allocation in which every agent receives a single item, while maximizing some objective. Typically, as well as in this paper, this objective is the (utilitarian) social welfare, i.e., the total utility of the agents. Since its introduction by \citet{hylland1979efficient}, this has been one of the most fundamental problems in the literature of economics (e.g., see \citep{BM:01,svensson1999strategy}), and has also been extensively studied in computational social choice (e.g., see \citep{Klaus2016matching}).

The classic work on the problem (including Hylland and Zeckhauser's seminal paper) assumes that the preferences of the agents are captured by cardinal valuation functions, assigning numerical values to the different items; these can be interpreted as their \emph{von Neuman-Morgenstern utilities} \citep{vnm}. From a more algorithmic viewpoint, one can envision a weighted complete bipartite graph (with agents and items forming the two sides of the partition), where the weights of the edges are given by these values. 
Crucially, most of the related literature assumes that the designer only has access to the preference rankings of the agents over the items (i.e., the \emph{ordinal preferences}) induced by the underlying values, but not to the values themselves.\footnote{\ The pseudo-market mechanism of \citet{hylland1979efficient} is a notable exception to this.} 
This is motivated by the fact that it is fairly standard to ask the agents to simply order the items, while it is arguably much more demanding to require them to specify exact numerical values for all of them. 

This begs the following natural question: \emph{What is the effect of this limited information on the goals of the algorithm designer?}  In 2006, \citeauthor{procaccia2006distortion} defined the notion of \emph{distortion} to measure precisely this effect, when the goal is to maximize the social welfare. Their original research agenda was put forward for settings in general social choice (also referred to as  \emph{voting}), but has since then flourished to capture several different scenarios, including the one-sided matching problem. 
  For the latter problem, \citet{Aris14}, showed that the best possible distortion achieved by any ordinal algorithm is $\Theta(\sqrt{n})$, even if one allows randomization, and even if the valuations are normalized. For deterministic algorithms,  the corresponding bound is $\Theta(n^2)$ (\cref{thm:ordinal}).

While the aforementioned bounds establish a stark impossibility when one has access only to ordinal information, they do not rule out the prospect of good approximations when it is possible to elicit  \emph{some} cardinal information. Indeed, the cognitive burden of eliciting cardinal values in the literature has mostly been considered in the two extremes; either full cardinal information or not at all. Conceivably though, if the agents needed to come up with only \emph{a few} cardinal values, the elicitation process would not be very demanding, while it could potentially have wondrous effects on the social welfare. This approach was advocated recently by \citet{amanatidis2020peeking}, who proposed to study the \emph{tradeoffs} between the number of cardinal \emph{value queries} per agent and distortion. For the general social choice setting of \citet{procaccia2006distortion}, \citet{amanatidis2020peeking} actually showed that with  a limited number of such queries, one can significantly improve upon the existing strong impossibilities \citep{boutilier2015optimal,caragiannis2017subset}.
Motivated by the success of this approach for general social choice settings, we extend this research agenda and aim to answer the following question for the one-sided matching problem:

\begin{quote}
\emph{What are the best possible information-distortion tradeoffs in one-sided matching? Can we achieve significant improvements over the case of only ordinal preferences, by making only a few cardinal value queries per agent?}
\end{quote}

\subsection{Our Contribution}

We consider the one-sided matching problem with the goal of maximizing the social welfare under limited information. We adopt the standard assumption in the related literature that the agents provide as input their ordinal preferences over the items, and that these are induced by their cardinal valuation functions. Following the agenda put forward by \citet{amanatidis2020peeking}, we also assume implicit access to the numerical values of the agents via \emph{value queries}; we may ask for an agent $i$ and an item $j$, and obtain the agent's value, $v_i(j)$, for that item. 

We measure the performance of an algorithm by the standard notion of \emph{distortion}, and our goal is to explore the tradeoffs between distortion and the number of queries we need per agent. As the two extremes, we note that if we use $n$ queries per agent, we recover the complete cardinal valuation profile and thus the distortion is $1$, whereas if we use $0$ queries, i.e., we use only the ordinal information, the best possible distortion is $\Theta(n^2)$ (see \cref{thm:ordinal}). The latter bound  holds even if we consider valuation functions that satisfy the \emph{unit-sum} normalization, i.e., the sum of the values of each agent for all the items is $1$. 
As we mentioned earlier, even when allowing randomization, the best possible distortion is still quite large ($\Theta(\sqrt{n})$ \citep{Aris14}) without employing any value queries. In this work, we only consider deterministic algorithms, 
and leave the study of randomized algorithms for future work.
\medskip

\noindent We provide the following  results: 
\begin{itemize}[leftmargin=15pt,itemsep=6pt,topsep=6pt,parsep=0pt,partopsep=0pt]
	\item In Section~\ref{sec:upper}, we present an algorithm parametrized by $\lambda$, which achieves distortion $O(n^{1/(\lambda+1)})$ by making $O(\lambda \log{n})$ queries per agent. In particular, by setting $\lambda=O(1)$ and  $\lambda=O(\log{n})$ we achieve respectively
	\begin{itemize}[itemsep=4pt,topsep=4pt,parsep=0pt,partopsep=0pt]
		\item distortion $O(\sqrt{n})$ using $O(\log{n})$ queries per agent;
		\item \emph{constant} distortion using $O(\log^2{n})$ queries per agent.
	\end{itemize}
	\smallskip
	The algorithm is inspired by a conceptually similar idea presented by \citet{amanatidis2020peeking} for the social choice setting. In Section~\ref{sec:extensions} we adapt our algorithm to provide analogous information-distortion tradeoffs for a wide range of well-studied optimization problems, including \emph{two-sided matching}, \emph{general graph matching} and the \emph{clearing problem for kidney exchange}.
	
	\item Next, still in \cref{sec:upper}, motivated by the analysis of the class of algorithms mentioned above as well as our lower bounds in \cref{sec:lower}, we consider a  class of instances (coined \emph{$k$-well-structured instances}) that captures the case where the agents (roughly) agree on the ranking of the items.
	We present a simple algorithm achieving distortion $O(k \cdot n^{1/k})$ by making only $k$ queries per agent for these instances. 
	
	
	\item In Section~\ref{sec:lower} we show a lower bound of $\Omega(n^{1/k}/k)$ on the distortion of any algorithm that makes $k$ queries per agent. An immediate consequence of this bound is that it is impossible to achieve constant distortion without asking almost $\log{n}$ queries! When $k$ is a constant, our aforementioned results on $k$-well-structured instances also establish the tightness of our construction, since the proof uses instances of this type. 
	Furthermore, using a construction which exploits the same ordinal but different cardinal information, we show that even under the stronger assumption of unit-sum normalization, the distortion cannot be better than  $\Omega\left(n^{1/(k+1)}/k\right)$ with $k$ queries per agent.
	
	\item In Section~\ref{sec:unit-sum} we present our main algorithmic result for unit-sum valuations, namely a novel algorithm which achieves distortion $O(n^{2/3}\sqrt{\log{n}})$ using only \emph{two} queries per agent. 
	
\end{itemize}
Our results are summarized in Figure~\ref{fig:results}. We note that our upper bounds are robust to ``errors'' in the responses to the queries (albeit not stated this way for the sake of simplicity). As long as the reported values are within a constant multiplicative factor from the true values, qualitatively there is no change in any of our bounds.

\tikzset{every picture/.style={line width=0.75pt}} 
\begin{figure*}
\small
	\centering
	
	\tikzset{every picture/.style={line width=0.75pt}} 
	
\begin{tikzpicture}[x=0.75pt,y=0.75pt,yscale=-0.78,xscale=0.86,every node/.style={scale=0.9}]

\tikzset{every picture/.style={line width=0.75pt}} 

\draw    (0,231) -- (715,231) ;
\draw [shift={(717,231.5)}, rotate = 180.04] [color={rgb, 255:red, 0; green, 0; blue, 0 }  ][line width=0.75]    (10.93,-3.29) .. controls (6.95,-1.4) and (3.31,-0.3) .. (0,0) .. controls (3.31,0.3) and (6.95,1.4) .. (10.93,3.29)   ;
\draw    (115,222) -- (115,237.5) ;
\draw    (682,222) -- (682,237.5) ;
\draw    (205,222) -- (205,237.5) ;
\draw    (535,222) -- (535,237.5) ;
\draw    (616,222) -- (616,237.5) ;
\draw    (310,222) -- (310,237.5) ;
\draw    (425,222) -- (425,237.5) ;
\draw  [color={rgb, 255:red, 74; green, 74; blue, 74 }  ,draw opacity=1 ][fill=white!80!gray  ,fill opacity=0.65 ] (0.5,136) -- (722,136) -- (722,220) -- (0.5,220) -- cycle ;
\draw  [color={rgb, 255:red, 74; green, 74; blue, 74 }  ,draw opacity=1 ][fill=white!70!orange  ,fill opacity=0.5 ] (0.5,58) -- (722,58) -- (722,136) -- (0.5,136) -- cycle ;

\draw (660,246) node [anchor=north west][inner sep=0.75pt]   [align=left] {\,{\footnotesize f}{\small ull info,}\\{\small $n$ queries}};
\draw (83,246) node [anchor=north west][inner sep=0.75pt]   [align=left] {{\small ordinal info,}\\{\small\ \   $0$ queries}};
\draw (96,87) node [anchor=north west][inner sep=0.75pt]    {$\Theta \left( n^{2}\right)$};
\draw (677,174.4) node [anchor=north west][inner sep=0.75pt]    {$1$};
\draw (180,246) node [anchor=north west][inner sep=0.75pt]   [align=left] {{\small  $2$ queries}};
\draw (180,69.4) node [anchor=north west][inner sep=0.75pt]    {$\tilde{O}\left( n^{2/3}\right)$};
\draw (180,148.4) node [anchor=north west][inner sep=0.75pt]    {$O\left(\sqrt{n}\right)$};
\draw (225,155) node [anchor=north west][inner sep=0.75pt]   [align=left] {\textcolor{gray}{\scriptsize [$2$-WS]}};
\draw (180,185) node [anchor=north west][inner sep=0.75pt]    {$\Omega \left(\sqrt{n}\right)$};
\draw (516,265) node [anchor=north west][inner sep=0.75pt]   [align=left] {{\small  queries}};
\draw (504,245) node [anchor=north west][inner sep=0.75pt]    {$O(\lambda \, \log n)$};
\draw (597,265) node [anchor=north west][inner sep=0.75pt]   [align=left] {{\small  queries}};
\draw (587,245) node [anchor=north west][inner sep=0.75pt]    {$O(\log^{2} n)$};
\draw (280,246) node [anchor=north west][inner sep=0.75pt]   [align=left] {{\small \ \ $k$ queries}\\{\small ($k = \Theta(1)$)}};
\draw (395,246) node [anchor=north west][inner sep=0.75pt]   [align=left] {{\small \ \ $k$ queries}\\{\small ($k = \omega(1)$)}};
\draw (396,181) node [anchor=north west][inner sep=0.75pt]    {$\Omega\big( \frac{n^{1/k}}{k}\big)$};
\draw (387,85) node [anchor=north west][inner sep=0.75pt]    {$\Omega\big( \frac{n^{1/(k+1)}}{k}\big)$};
\draw (602,174.4) node [anchor=north west][inner sep=0.75pt]    {$O( 1)$};
\draw (496,172) node [anchor=north west][inner sep=0.75pt]    {$O\left( n^{1/( \lambda +1)}\right)$};
\draw (285,184) node [anchor=north west][inner sep=0.75pt]    {$\Omega \left( n^{1/k}\right)$};
\draw (274,87.4) node [anchor=north west][inner sep=0.75pt]    {$\Omega \left( n^{1/( k+1)}\right)$};
\draw (285,148.4) node [anchor=north west][inner sep=0.75pt]    {$O \left( n^{1/k}\right)$};
\draw (338,155) node [anchor=north west][inner sep=0.75pt]   [align=left] {\textcolor{gray}{\scriptsize [$k$-WS]}};
\draw (390,148.4) node [anchor=north west][inner sep=0.75pt]    {$O \left( k \, n^{1/k}\right)$};
\draw (453,155) node [anchor=north west][inner sep=0.75pt]   [align=left] {\textcolor{gray}{\scriptsize [$k$-WS]}};
\draw (88,174) node [anchor=north west][inner sep=0.75pt]  [font=\small] [align=left] {unbounded};
\draw (4,166) node [anchor=north west][inner sep=0.75pt]  [font=\small,color={rgb, 255:red, 0; green, 0; blue, 0 }  ,opacity=1 ] [align=left] {\textbf{\textit{unrestricted }}\\\textbf{\textit{valuations}}};
\draw (180,102.4) node [anchor=north west][inner sep=0.75pt]    {$\Omega \left( n^{1/3}\right)$};
\draw (9,88) node [anchor=north west][inner sep=0.75pt]  [font=\small,color={rgb, 255:red, 0; green, 0; blue, 0 }  ,opacity=1 ] [align=left] {\textbf{\textit{unit-sum}}\\\textbf{\textit{valuations}}};
\end{tikzpicture}
	
	\caption{An overview of our results. Some of the distortion guarantees hold for $k$-well-structured instances; this is indicated in brackets next to the bound. The upper bounds that hold for unrestricted valuations also obviously hold for unit-sum valuations. All of the lower bounds hold even for instances with the same ordinal preferences for all agents. 
	}
	\label{fig:results}
\end{figure*}

\subsubsection*{Technical Overview}
Our parametrized class of algorithms developed in \cref{sec:upper} is based on the following idea. 
First, every agent is queried for their favorite item. Then, for each agent, the algorithm partitions the items into sets so that the value of the agent for all items in a set is lower-bounded by a carefully defined quantity; these sets are constructed via a sequence of binary search subroutines. Finally, the algorithm outputs a matching that maximizes the social welfare, with respect to the ``simulated'' values obtained via this process.

A similar idea was proposed by \cite{amanatidis2020peeking} for the social choice setting. We remark, however, that if one translates their result directly to our setting (by interpreting matchings as \emph{alternatives} and running their algorithm), it yields inconsequential distortion bounds, as well as an exponential-time running time for the algorithm. To obtain meaningful bounds, the key is to adopt the \emph{principle} of their approach rather than the exact solution proposed. As a matter of fact, in \cref{sec:extensions}, we show that the same principle can be further refined and applied to a plethora of other combinatorial optimization problems on graphs with additive objectives.  
In particular, we show bounds of similar flavor for several well-studied problems, such as
\begin{itemize}[leftmargin=15pt,itemsep=6pt,topsep=6pt,parsep=0pt,partopsep=0pt]
	\item two-sided (perfect) matching (i.e., a (perfect) matching on a bipartite graph with agents on each side);
	\item general graph matching (i.e., a matching on a general graph with agents being the vertices);
	\item general resource allocation (i.e., the allocation of $m$ items to $n$ agents under various constraints);
	\item max $k$-sum clustering (i.e., a generalization of matchings on graphs, see \citep{anshelevich2016blind});
	\item the clearing problem for kidney exchanges (i.e., a cycle cover with restricted cycle length, see \citep{abraham2007clearing}).
\end{itemize}
With regard to our results for $k$-well-structured instances in \cref{sec:upper}, we note that we obtain a notable improvement over the tradeoffs achieved by the aforementioned approach. For instance, a distortion of $O(\sqrt{n})$ is achievable here using only {\em two} queries! To provide some intuition, one can think of $k$-well-structured instances as capturing cases for which there is  a general agreement on the set of the ``most valuable items'', although the agents might rank the items in this set in different ways. 
For example, most researchers in Artificial Intelligence would agree on the top 5 publication venues, although they might rank those 5 venues differently. The parameter $k$ captures the different ``levels'' of agreement.

An interesting class of instances that are $k$-well-structured for every $k$ is that of instances where all the agents have the \emph{same ranking} over all items. These instances are important because intuitively, they highlight the challenge of social welfare maximization under ordinal information. How is an algorithm in that case supposed to distinguish between pairs of high  and low value? Perhaps somewhat surprisingly, it turns out that such instances are more amenable to handling via a smaller number of queries.
In fact, the related literature has been concerned in the past with this type of instances; e.g., \cite{Aris14} used such instances in their lower bound constructions and referred to them as {\em ordered} instances, and \cite{plaut2020almost} and \cite{barman2020approximation} considered such instances in the context of fair division of indivisible items. We also use such instances for our lower bound constructions in \cref{sec:lower}; in that sense, we do not only provide improved upper bounds for an interesting class of instances but we also show the tightness of the analysis for our lower bound constructions.  

We remark that the results of \cref{sec:upper} and \cref{sec:extensions} do not require any normalization assumptions. For agents with unit-sum normalized valuations, in \cref{sec:unit-sum}, we present an algorithm which achieves a distortion of $O(n^{2/3}\sqrt{\log{n}})$ using only two queries per agent. The algorithm is adaptive, and uses the second query differently, depending on the maximum value that it sees after querying all agents at the first position. While this result is not tight based on our lower bound of $\Omega(n^{1/3})$ for this case, we consider it as one of the highlights of this work. In particular, it shows that sublinear distortion is possible even with a deterministic algorithm using a constant number of queries per agent. 

\subsection{Related Work}
The one-sided matching problem in the context of agents with preferences over items was firstly introduced by \cite{hylland1979efficient}. The classic literature in economics (e.g., see \citep{BM:01, svensson1999strategy} and references therein) was mostly concerned with axiomatic properties, and has proposed several solutions and impossibilities; see the surveys of \cite{sonmez2011matching} and \cite{abdulkadiroglu2013matching} for more information. 

The effects of limited information on the social welfare objective were studied most notably in the work of \cite{Aris14} mentioned earlier. Further, a line of work \citep{anshelevich2016blind,anshelevich2016truthful,anshelevich2017tradeoffs,abramowitz2017utilitarians} studied related settings on graphs, and showed distortion bounds for matching problems and their generalizations. 
A crucial difference from our work is that they consider \emph{symmetric} weights on the edges of the graph, which  corresponds to cases where agents are paired with other agents (e.g., matching or clustering) and the value of an agent $i$ for another agent $j$ is the same as the value of $j$ for $i$. In contrast, in the graph problems that we consider, the weights are assumed to be \emph{asymmetric}; the weight of an edge is given by the sum of weights of incident vertices.\footnote{\ \cite{anshelevich2013anarchy} refer to this setting as ``Asymmetric Edge-Labeled Graphs'' as opposed to ``Symmetric Edge-Labeled Graphs'', which is the setting of \citep{anshelevich2016blind} and the other works mentioned above.} This makes the results markedly different. Another important distinction is that most of the aforementioned works operate in the setting where the edge weights satisfy the triangle inequality, whereas we impose no such restriction. \cite{caragiannis2016truthful} study one-side matching settings in metric spaces, and thus their work also falls into this category, but is quite distinct as they focus on cost objectives rather than welfare.

For general social choice settings (i.e., voting), the distortion of ordinal algorithms has been studied in a long list of papers, e.g., see \citep{procaccia2006distortion,anshelevich2017randomized,anshelevich2018approximating,boutilier2015optimal,caragiannis2017subset,benade2017preference,caragiannis2018truthful,fain2018random,filos2014truthful,goel2017metric,munagala2019improved,feldman2016facility,gkatzelis2020resolving}. Most of the related work considers the standard case where only ordinal information is given, with a few notable exceptions \citep{abramowitz2019awareness,benade2017preference,bhaskar2018truthful,filos2020distortion,filos2020approximate}. 

The approach of enhancing the expressiveness of algorithms by equipping them with cardinal queries that we adopt in this paper was first suggested by \cite{amanatidis2020peeking}. It should be noted that this is in nature quite different from another related recent approach proposed by \cite{mandalefficient, mandal2020optimal}, which considers the communication complexity of voting algorithms. In that setting, the algorithm must elicit a limited number of \emph{bits} of information from the agents and it is not assumed that the ordinal preferences are already known. Moreover, the agents are allowed within that number of bits to communicate partial information about all of their different values. 

\section{Model definition}\label{sec:prelim}
We consider the one-sided matching problem, where there is a set of agents $N$ and a set of items $A$, such that $|N|=|A|=n$. Each agent $i \in N$ has a {\em valuation function} $v_i: A \rightarrow \RR_{\geq 0}$ indicating the agent's value for each item; that is $v_i(j)$ is the value of agent $i \in N$ for item $j \in A$. 
The valuation functions we consider are either
\begin{itemize}[itemsep=6pt,topsep=6pt,parsep=0pt,partopsep=0pt]
\item {\em unrestricted}, in which case the values for the items can be any non-negative real numbers, or
\item {\em unit-sum}, in which case the sum of values of each agent $i$ for all items is $1$: $\sum_{j \in A}v_i(j)=1$. 
\end{itemize}
We denote by $\vv=(v_i)_{i \in N}$ the {\em (cardinal) valuation profile} of the agents.  
Let $Y=(y_i)_{i \in N}$ be a {\em matching} according to which each agent $i \in N$ is matched to exactly one item $y_i \in A$, such that $y_i \neq y_{i'}$ for every $i \neq i'$. Given a profile $\vv$, the {\em social welfare} of $Y$, $\SW(Y|\vv)$, is the total value of the agents for the items they are matched to according to $Y$:
\begin{align*}
\SW(Y|\vv) = \sum_{i \in N} v_i(y_i)\,.
\end{align*}
By $\mathcal{M}$ we denote the set of all perfect matchings on our instance. Our goal is to compute a matching $X(\vv) = (x_i)_{i \in N}$ with maximum social welfare, i.e., 
\begin{align*}
X(\vv) \in \arg\max_{Y\in \mathcal{M}} \SW(Y|\vv)\,.
\end{align*}

In case the valuation functions of the agents are known, then computing $X(\vv)$ can be done efficiently, e.g., via the Hungarian method \citep{Kuhn56}.
However, our setting is a bit more restrictive. The exact valuation functions of the agents are their private information, and they can instead report orderings over the items, which are consistent with their valuations. In particular, every agent $i$ reports a ranking of the items $\succ_i$ such that 
$a \succ_i b$ if and only if $v_i(a) \geq v_i(b)$ for all $a,b \in A$. Given a valuation profile $\vv$, we denote by $\succ_\vv = (\succ_i)_{i \in N}$ the {\em ordinal profile} induced by $\vv$; observe that different valuation profiles may induce the same ordinal profile. On top of the ordinal preferences of the agents, we can obtain partial access to the valuation profile, by making a number of value queries per agent. In particular, a {\em value query} takes as input an agent $i \in N$ and an item $j \in A$, and returns the value $v_i(j)$ of agent $i$ for item $j$. This leads us to the following definition of a deterministic algorithm in our setting.

\begin{definition}
A {\em matching algorithm} $\calA_k$ takes as input an ordinal profile $(\succ)_{i \in N}$, makes $k\le n$ value queries per agent and, using $(\succ)_{i \in N}$ as well as the answers to the queries, it computes a matching $\calA_k(\succ) \in \mathcal{M}$. If $k=0$,  $\calA$ is an {\em ordinal} algorithm, whereas if $k=n$, $\calA$ is a {\em cardinal} algorithm.
\end{definition}

As already mentioned, we can efficiently compute the optimal matching using a cardinal algorithm. However, if an algorithm is allowed to make a limited number $k < n$ of queries per agent, the computed matching might not be optimal. The question then is how well does such an algorithm {\em approximate} the optimal social welfare of any matching. Approximation here is captured by the notion of {\em distortion}.

\begin{definition}
The {\em distortion} $\dist(\calA_k)$ of an algorithm $\calA_k$ is the worst-case ratio (over all possible valuation profiles $\vv$) between the social welfare of an optimal matching $X(\vv)$ and the social welfare of the matching computed by $\calA_k$:
\begin{align*}
\dist(\calA_k) = \sup_{\vv} \frac{\SW(X(\vv)|\vv)}{\SW(\calA_k(\succ_\vv)|\vv)}\,.
\end{align*}
\end{definition}

\subsection{Warm-up: Ordinal Algorithms}\label{sec:ordinal}
Before we proceed with our more technical results on tradeoffs between information and distortion, we consider the case of ordinal algorithms. When the valuation functions of the agents are unrestricted, the distortion of any ordinal algorithm is unbounded. To see this, consider any instance that contains two agents who agree on which the most valuable item is.
Since only one of them can be matched to this item, it might be the case that the other agent has an arbitrarily large value for it, leading to unbounded distortion. Even for the more restrictive case of unit-sum valuations, however, the distortion of ordinal algorithms can be quite large.

\begin{theorem}\label{thm:ordinal}
The distortion of the \emph{best} ordinal matching algorithm is $\Theta(n^2)$.
\end{theorem}

\begin{proof}
For the upper bound, consider any algorithm that outputs a matching so that some agent is matched to her top-ranked item. 
As the valuations are unit-sum, this agent must have value at least $1/n$ for this item, and thus the social welfare of the matching computed by the algorithm is in turn also at least $1/n$. Since the value of every agent for any item is at most $1$, the maximum possible social welfare is upper-bounded by $n$, and thus the distortion of the algorithm is at most $n^2$.

For the lower bound, we assume that $n$ is even; our instance can be easily adjusted for odd $n$. We consider an instance with set of items $A = \{a, b_1, \ldots, b_{n/2}, c_1, \ldots, c_{n/2-1}\}$. 
The ordinal profile is such that, for $i\in \{1,\ldots,n/2\}$, agents $i$ and $i+n/2$ have the same ordinal preference $\succ_i$, defined as 
\[
a \succ_i b_i \succ_i b_1 \succ_i \ldots \succ_i b_{i-1} \succ_i b_{i+1} \succ_i \ldots \succ_i b_{n/2} \succ_i c_1 \succ_i \ldots \succ_i c_{n/2-1}\,.
\]
Consider any ordinal algorithm which, given as input this profile, outputs a matching $Y = (y_i)_{i \in N}$. 
We define a valuation profile $\vv$ which is consistent with the ordinal profile, and the value of agent $i \in \{1,\ldots, n\}$ depends on the structure of $Y$. For convenience, let $s_i$ denote the second favorite item of agent $i$, i.e., $s_i = b_{i}$ if $i\le n/2$ and $s_i = b_{i-n/2}$ if $i> n/2$.
\begin{itemize}[itemsep=6pt,topsep=6pt,parsep=0pt,partopsep=0pt]
\item If $y_i = a$, then agent $i$ has value $1/n$ for all items;
\item If $y_i = s_{i}$, then agent $i$ has value $1$ for $a$ and $0$ for every other item;
\item Otherwise, agent $i$ has value $1/2$ for $a$, $1/2$ for $s_{i}$ and $0$ for every other item.
\end{itemize}
Since only one agent can be matched to $a$, and everyone else will be matched to an item of value $0$, 
the social welfare of $Y$ computed by the algorithm is $\SW(Y | \vv) = 1/n$. 
However, observe that there exists a matching $X$ with social welfare $\SW(X | \vv) \approx n/4$. 
In particular, we can go through the agents and match each agent $i \in \{1,\ldots, n\}$ to $s_{i}$ if she is not already matched to $s_i$ or to $a$ in $Y$. This way we will end up with a matching where at least $n/2-1$ agents will have value $1/2$ each for the corresponding items. Our claim about the social welfare follows. Therefore, the distortion of any ordinal algorithm is $\Omega(n^2)$.
\end{proof}

\section{Distortion Guarantees for Unconstrained Valuations}\label{sec:upper}
Here we present $\lambda$-$\gaga$ ($\lambda$-TSF), 
an algorithm that works for any valuation functions. At a high level, for each agent, we do the following. We first query the agent's value for her highest ranked item. Then, we partition the items into $\lambda + 1$ sets, so that the agent's value for all the items in a set is lower-bounded by a carefully selected quantity related to the agent's top value. Based on this partition, we then define a new \emph{simulated valuation function} for the agent, where the value of an item is equal to the lower bound that corresponds to the set the item belongs to. Finally, we compute a maximum weight matching with respect to the simulated valuation functions. Formally:
{
\begin{mdframed}[backgroundcolor=white!90!gray,
	leftmargin=\dimexpr\leftmargin-20pt\relax,
	rightmargin=\dimexpr\rightmargin+5pt\relax
	skipabove=5pt,skipbelow=5pt]
	
	$\lambda$-$\gaga$ ($\lambda$-TSF)\\[-5pt]

\noindent Let 
$\alpha_\ell = n^{-\ell/(\lambda+1)}$ for $\ell \in \{0, ..., \lambda\}$. \\[-5pt]

\noindent For every agent $i \in N$:
\begin{itemize}[itemsep=6pt,topsep=6pt,parsep=0pt,partopsep=0pt]
\item Query $i$ for her top-ranked item $j_i^*$; let $v_i^*$ be this value. 

\item Let $Q_{i,0} = \{j_i^*\}$ and $\tilde{v}_i(j_i^*) = \alpha_0 \cdot v_i^* = v_i^*$. 

\item For every $\ell \in \{1, ..., \lambda\}$, using binary search, compute  
\begin{align*}
Q_{i,\ell} = \left\{j \in A: v_i(j) \in \left[ \alpha_\ell \cdot v_i^*, \alpha_{\ell-1} \cdot v_i^* \right) \right\}
\end{align*}
and let $\tilde{v}_i(j) = \alpha_\ell \cdot v_i^*$ for every $j \in Q_{i,\ell}$.

\item Let $Q_i = \bigcup_{\ell=0}^{\lambda} Q_{i,\ell}$ and set $\tilde{v}_i(j) = 0$ for $j \in A \setminus Q_i$.
\end{itemize}

\noindent Return a matching $Y \in \arg\max_{Z\in\mathcal{M}} \SW(Z | \tilde{\vv})$.
\end{mdframed}}
\smallskip
The next theorem follows (asymptotically) from the more general \cref{thm:binary-general}, which is stated in \cref{sec:extensions} and applies to a number of well-known graph problems. To aid the reader, we also include a self-contained, cleaner proof of \cref{thm:lambdaTSFdistortion} which applies only to one-sided matching and gives a slightly better bound.

\begin{theorem}\label{thm:lambdaTSFdistortion}
$\lambda$-TSF makes $1 + \lambda + \lambda \log{n} $ queries per agent and achieves a distortion of $2 n^{1/(\lambda+1)} $.
\end{theorem}

\begin{proof}
Consider an arbitrary valuation profile $\vv$. Let $X=X(\vv)$ be an optimal matching, and $Y$ be the solution returned by $\lambda$-TSF; recall that $Y$ maximizes the social welfare according to the simulated valuation functions $\tilde{v}_i$. Let $S$ be the set of agents $i$ such that $x_i \in Q_i$, and $\overline{S} = N \setminus S$. 
Then, the optimal social welfare can be written as 
\begin{align*}
\SW(X|\vv) &= \sum_{i \in N} v_i(x_i)
= \sum_{i \in \overline{S}} v_i(x_i) + \sum_{i \in S} v_i(x_i).
\end{align*}

We will bound the two terms separately. We begin with the first one:
\begin{align*}
\sum_{i \in \overline{S}} v_i(x_i) 
&<  \alpha_{\lambda}  \sum_{i \in \overline{S}}  v_i^*  
\leq \alpha_{\lambda}  \cdot n \cdot \max_{i \in N}  v_i^* 
\leq \alpha_\lambda \cdot n \cdot \SW(Y|\tilde{\vv}) 
\leq \alpha_\lambda \cdot n \cdot \SW(Y|\vv). 
\end{align*}
The first inequality follows directly by the definition of $Q_i$. 
The second inequality follows since $\overline{S} \subseteq N$.
The third inequality follows since $Y$ maximizes the social welfare according to the profile $\vv$ and the algorithm has queried all agents for their most-preferred items. Finally, the last inequality follows since the simulated values of an agents are lower bounds on her true ones. 

For the second term, let $S_\ell$ be the restriction of $S$ on agents for whom $x_i \in Q_{i,\ell}$, $\ell \in \{0, ..., \lambda\}$. Then,
\begin{align*}
\sum_{i \in S} v_i(x_i)
= \sum_{\ell = 0}^\lambda \sum_{i \in S_\ell} v_i(x_i) \,.
\end{align*}
Now, let us assume that $\lambda>0$; we will deal with the simpler case where $\lambda=0$ later.
By definition, for any $\ell \in \{1, \ldots, \lambda\}$ and any $j \in Q_{i,\ell}$, we have that $v_{i}(j) \leq \alpha_{\ell-1} \cdot v_i^* = \frac{\alpha_{\ell-1}}{\alpha_{\ell}} \cdot \alpha_\ell \cdot v_i^* =  \tilde{v}_{i}(j) / \alpha_1$.  Also, for $Q_{i,0} = \{j_i^*\}$, we have $v_{i}(j_i^*) = \tilde{v}_{i}(j_i^*) \leq \tilde{v}_{i}(j_i^*) / \alpha_1$. Hence, 
\begin{align*}
\sum_{i \in S} v_i(x_i)
&\leq \alpha_1^{-1} \sum_{\ell = 0}^\lambda \sum_{i \in S_\ell} \tilde{v}_i(x_i) 
\leq \alpha_1^{-1} \sum_{i \in N} \tilde{v}_i(x_i) 
\leq \alpha_1^{-1}  \SW(Y|\tilde{\vv}) 
\leq \alpha_1^{-1}  \SW(Y|\vv) \,.
\end{align*}
The second inequality follows by considering all agents. The third inequality follows from the optimality of $Y$ with respect to the simulated valuation functions. Finally, the last inequality follows follows since the simulated values of an agents are lower bounds on her true ones. 

Now we can put everything together:
\begin{align} \label{eq:gaga-bound}
\SW(X|\vv)
\leq \left( \alpha_{\lambda}\cdot n  +  \alpha_1^{-1} \right) \SW(Y|\vv)
= 2n^{1/(\lambda+1)} \cdot  \SW(Y|\vv)  \,, 
\end{align}
and this settles the bound on the distortion when $\lambda>0$. 

When $\lambda=0$, we clearly have that
\[\sum_{i \in S} v_i(x_i) \leq \sum_{i \in S} \tilde{v}_i(x_i) \leq \SW(Y|\vv)\,.\]
Then, the analog of \eqref{eq:gaga-bound} is 
\begin{align*}
\SW(X|\vv)
&\leq \left( \alpha_{\lambda}\cdot n  +  1 \right) \SW(Y|\vv) \nonumber 
= 2n^{1/(\lambda+1)} \cdot  \SW(Y|\vv) \,.
\end{align*}
This concludes the proof.
\end{proof}

By appropriately setting the value of $\lambda$, we obtain several tradeoffs between the distortion and the number of queries per agent. In particular, we have the following statement.

\begin{corollary}
We can achieve
\begin{itemize}[itemsep=6pt,topsep=6pt,parsep=0pt,partopsep=0pt]
\item distortion $O(n)$ by making \emph{one} query per agent;
\item distortion $O(n^{1/k})$ for any constant integer $k$ by making $O(\log{n})$ queries per agent;
\item distortion $O(1)$ by making $O(\log^2{n})$ queries per agent.
\end{itemize}
\end{corollary}

\subsection{Well-structured instances}\label{sec:well-structured}

We now consider instances in which the agents exhibit quite similar ordinal preferences.
For any positive integer $k$, we define the class of {\em $k$-well-structured} ($k$-WS) instances. 
Let $\varepsilon \in (0,1)$ be a constant (e.g., $\varepsilon=1/2$). 
In a $k$-WS instance, the set of items can be partitioned into $k+1$ sets $A_1$, \ldots, $A_k$, $A_{k+1}$ such that 
\begin{align*}
|A_1|=1 \text{ \ \ and \ \ } |A_\ell|=\varepsilon \cdot n^{(\ell-1)/k} \text{ for all } \ell \in \{2, \ldots, k\}\,,
\end{align*}
and every agent $i$ has the ordinal preference
\begin{align*}
\langle A_1 \rangle_i \succ_i \langle A_2 \rangle_i \succ_i \ldots \succ_i \langle A_k \rangle_i \succ_i \langle A_{k+1} \rangle_i\,,
\end{align*}
where $\langle A_\ell \rangle_i$ denotes some ordering of the items in set $A_\ell$ which depends on agent $i$; that is, different agents may order the items in $A_\ell$ differently. Observe that an instance in which all agents have the same ranking over the items is a $k$-WS instance for every $k$. We will use such instances in our lower bounds in the next section.

For the class of $k$-WS instances we present a simple algorithm, names {\sc $k$-FixedMaxMatching} ($k$-FMM), which achieves a distortion of $O(k \cdot n^{1/k})$ by making $k$ queries per agent. 
It works as follows:
\smallskip

\begin{mdframed}[backgroundcolor=white!90!gray,
	leftmargin=\dimexpr\leftmargin-20pt\relax,
	rightmargin=\dimexpr\rightmargin+5pt\relax
	skipabove=5pt,skipbelow=5pt]
{\sc $k$-FixedMaxMatching} ($k$-FMM)\\[-5pt]

\noindent For every $\ell \in \{1,\ldots,k\}$, query the value of each agent $i$ for her least-preferred item in $A_\ell$; denote this value by $u_i(\ell)$.\\[-5pt]

\noindent For each $i\in N$, define the simulated valuation function $\tilde{v}_i$: 
\begin{itemize}[itemsep=6pt,topsep=6pt,parsep=0pt,partopsep=0pt]
\item $\tilde{v}_i(j) = u_i(\ell)$ for every $j \in A_\ell$, $\ell \in \{1,\ldots,k\}$;
\item $\tilde{v}_i(j) = 0$ for every $j \in A_{k+1}$.
\end{itemize}
\noindent Return a matching $Y \in \arg\max_Z \SW(Z | \tilde{\vv})$. 
\end{mdframed}
\smallskip

\begin{theorem}
For the class of $k$-well-structured instances with $k \geq 1$, {\sc $k$-FMM} makes $k$ queries per agent and achieves a distortion of $O(k \cdot n^{1/k})$.
\end{theorem}

\begin{proof}
Consider an arbitrary $k$-WS instance with valuation profile $\vv$. 
For $\ell \in \{1,\ldots,k\}$, denote by $S_\ell$ the set of $|A_\ell| = n^{(\ell-1)/k}$ agents with the highest values for the last item in $A_\ell$ (breaking ties arbitrarily). Since there exists a matching of the items in $A_\ell$ to the agents of $S_\ell$, the algorithm maximizes the simulated welfare, and $v_i(j) \geq \tilde{v}_i(j)$ for every agent $i$ and item $j$, the social welfare of the matching $Y$ computed by the algorithm is
\begin{align*}
\SW(Y| \vv) &\geq \sum_{i \in S_\ell} u_{i}(\ell) \\
&\geq n^{(\ell-1)/k} \cdot \min_{i \in S_\ell} u_i(\ell)\,,
\end{align*}
for every $\ell \in \{1,\ldots,k\}$. Observe that since $A_1$ consists of just one item, the inequality for $\ell=1$ can be simplified to 
\[\SW(Y|\vv) \geq \max_i u_i(1)\,.\]
Now, let $X$ be an optimal matching, and denote by $x_j$ the agent matched to item $j \in A$. Then,
\begin{align*}
\SW(X|\vv) 
&= \sum_{\ell=1}^{k+1}  \sum_{j \in A_\ell} v_{x_j}(j) 
=  \sum_{j \in A_1} v_{x_j}(j) +  \sum_{\ell=2}^{k+1} \sum_{j \in A_\ell} v_{x_j}(j) \\
&\leq \max_i u_i(1) +  \sum_{\ell=2}^{k+1} \sum_{j \in A_\ell} u_{x_j}(\ell-1)\,,
\end{align*}
where the inequality follows since $A_1$ consists of just one item, and the values of the agents for every item in $A_\ell$ are at most their values for the last item in $A_{\ell-1}$ (since all items in $A_{\ell-1}$ are ranked higher than the items in $A_\ell$). 
Let us focus on the right-most term of the above expression. We have
\begin{align*}
\sum_{\ell=2}^{k+1} \sum_{j \in A_\ell} u_{x_j}(\ell-1)  
&= \sum_{\ell=2}^{k+1} \sum_{j \in A_\ell} \sum_{x_j \in S_{\ell-1}} u_{x_j}(\ell-1) 
 + \sum_{\ell=2}^{k+1} \sum_{j \in A_\ell}  \sum_{x_j \not\in S_{\ell-1}} u_{x_j}(\ell-1) \,.
\end{align*}
By the definition of $S_{\ell-1}$, for every agent $i \not\in S_{\ell-1}$, it holds that $u_{x_j}(\ell-1) \leq \min_{i \in S_{\ell-1}} u_i(\ell-1)$, and thus
\begin{align*}
\sum_{\ell=2}^{k+1} \sum_{j \in A_\ell} u_{x_j}(\ell-1)  
&\leq  \sum_{\ell=2}^{k+1} \sum_{j \in A_\ell} \sum_{x_j \in S_{\ell-1}} u_{x_j}(\ell-1)
+ \sum_{\ell=2}^{k+1} \sum_{j \in A_\ell} \sum_{x_j \not\in S_{\ell-1}}\!\! \min_{i \in S_{\ell-1}}\!\! u_i(\ell-1)\,.
\end{align*}
Observe that the internal double sum of the first term can sum over at most all agents in $S_{\ell-1}$, while the internal double sum of the second term can sum over at most $|A_\ell|$ agents. Using the fact that $|A_{k+1}| \leq n$ and the lower bounds for the social welfare of the matching computed by the algorithm, we obtain 
\begin{align*}
\sum_{\ell=2}^{k+1} \sum_{j \in A_\ell} u_{x_j}(\ell-1)  
&\leq \sum_{\ell=2}^{k+1} \sum_{i \in S_{\ell-1}} u_i(\ell-1) 
+ \sum_{\ell=2}^{k+1} |A_\ell|  \min_{i \in S_{\ell-1}} u_i(\ell-1) \\
&= \sum_{\ell=2}^{k+1} \sum_{i \in S_{\ell-1}}\!\! u_i(\ell-1) 
 + \sum_{\ell=2}^{k} n^{(\ell-1)/k} \cdot  \min_{i \in S_{\ell-1}} u_i(\ell-1)  + |A_{k+1}| \min_{i \in S_k} u_i(k) \\
&\leq  \sum_{\ell=1}^k \sum_{i \in S_\ell}\! u_i(\ell) 
 + n^{1/k} \sum_{\ell=1}^{k-1} n^{(\ell-1)/k} \cdot \min_{i \in S_\ell} u_i(\ell) + n \cdot \min_{i \in S_k} u_i(k)\\
&\leq (k + k\cdot n^{1/k} ) \cdot \SW(Y | \vv)\,.
\end{align*}
Putting everything together, we obtain the theorem.
\end{proof}

We conclude our discussion on $k$-WS instances with a lower bound $\Omega(n^{1/k})$ on the distortion of the $(k-1)$-TSF algorithm; here, the value of the parameter $\lambda$ is chosen to be $k-1$ because of the structure of $k$-WS instances. Combined together with the $O(k\log{n})$ queries that it requires to operate, we have that when $k$ is sub-logarithmic, the $k$-FMM algorithm presented above matches the distortion of $(k-1)$-TSF on $k$-WS instances, using a factor of $\log{n}$ less queries per agent. 

\begin{theorem}\label{thm:gaga-tight}
For every constant $k \geq 1$, there exists a $k$-well-structured instance such that the distortion of $(k-1)$-TSF when given this instance as input is $\Omega(n^{1/k})$. 
\end{theorem}

\begin{proof}
Consider an instance in which all agents have the same ordinal preference over the items, which can be easily seen to be a $k$-WS instance for every $k \geq 1$. Let us now define the cardinal values which are revealed when the $(k-1)$-TSF algorithm queries the agents:
\begin{itemize}[itemsep=6pt,topsep=6pt,parsep=0pt,partopsep=0pt]
\item $\alpha_{\ell-1} = n^{-(\ell-1)/k}$ for queries about items in $A_\ell$, $\ell\le k$;
\item $0$ for queries about items in $A_{k+1}$. 
\end{itemize}
Because the rankings of the agents for the items and the revealed information due to the queries of the algorithm are the same among the agents, the algorithm will define the same simulated valuation function for all agents. In particular, based on the revealed values, we have that $Q_{i,\ell-1} = A_\ell$ for every $\ell \in \{1,\ldots,k\}$. By considering a valuation profile $\vv$ according to which an agent that is matched to an item in set $A_\ell$ has value $\alpha_{\ell-1}$ for it, we have that the social welfare of the matching $Y$ computed by the algorithm is
\begin{align*}
\SW(Y|\vv) 
&= \sum_{\ell=1}^{k}  |A_\ell| \cdot \alpha_{\ell-1} 
= 1 + \varepsilon \sum_{\ell=2}^k n^{(\ell-1)/k} \cdot  n^{-(\ell-1)/k} 
= 1 +\varepsilon \cdot (k-1) \leq  k\,.
\end{align*}
Now, observe that when binary search is restricted to run over only the items in the set $A_\ell$, $\ell \in \{2, \ldots, k+1\}$, it does not query all the items in $A_\ell$; in particular, because of the way binary search operates, the value of the first $|A_\ell|/2$ items therein will never be revealed. Hence, even if the algorithm matches the agents to items for which they have not been queried for, there must exist a matching $X$ such that for every $\ell \in \{2, \ldots, k+1\}$ the first $|A_\ell|/2$ items of $A_\ell$ are matched to agents different than the ones chosen by the algorithm, who have not been queried for their values. By setting the real values of the agents for these items to be $\alpha_{\ell-2}$, and observing that $|A_{k+1}|= \xi n$ for some constant $\xi  \in (0,1)$, we have
\begin{align*}
\SW(X|\vv) &\geq \sum_{\ell=2}^{k+1} \frac{|A_\ell|}{2}\cdot \alpha_{\ell-2}  
&= \frac{\varepsilon}{2}\sum_{\ell=2}^k n^{(\ell-1)/k} \cdot  n^{-(\ell-2)/k}
+ \frac{\xi}{2} \cdot n \cdot n^{-(k-1)/k} 
\geq  \frac{\min\{\varepsilon,\xi\}}{2}\cdot k \cdot n^{1/k}.
\end{align*}
Hence, the distortion is $\Omega(n^{1/k})$.
\end{proof}

\section{Lower bounds}\label{sec:lower}

In this section we show unconditional lower bounds for algorithms for one-sided matching which are allowed to make at most $k \geq 1$ queries per agent. We  
present a generic matching instance which can be fine-tuned to yield lower bounds for both unrestricted and unit-sum valuation functions. Let $\calV$ denote any of these two classes of valuation functions.  

Let $\delta_\calV(k) \leq 1/k$ be a function of $k$, and $\varepsilon \in (0,1/2)$ be some constant. 
We want to define an instance in which the $n$ items are partitioned into $k+2$ sets 
$A_1, ..., A_{k+1}, B = A \setminus \big( \bigcup_{\ell \in [k+1]}A_\ell \big)$
such that 
\[|A_\ell| = \varepsilon \cdot n^{(\ell-1) \delta_\calV(k)}.\]
Note that because we have restricted the possible values of $\delta_\calV(k)$ to be at most $1/k$ and have chosen $\varepsilon < 1$, these sets of items can be defined; in particular,  $|A_{k+1}|\leq \varepsilon n$. 
We assume that $n$ is large enough so that $n > 2\sum_{\ell=1}^{k+1} |A_\ell|$ and that is such that the cardinalities are indeed integers; the latter is only assumed to simplify the notation. 
We use $\langle T \rangle$ to denote some fixed arbitrary  ranking of the elements of set $T$ (which is common for all agents). Given that, we define the ordinal preference of every agent $i \in N$ to be
\[
\langle A_1 \rangle \succ_i \langle A_2 \rangle \succ_i ... \succ_i \langle A_k \rangle \succ_i \langle A_{k+1} \rangle \succ_i \langle B \rangle.
\]
We reveal the following information, depending on the queries of the algorithm:
\begin{itemize}[itemsep=6pt,topsep=6pt,parsep=0pt,partopsep=0pt]
\item For every $\ell \in \{1,\ldots,k+1\}$, any query for some item in $A_\ell$ reveals a value of $|A_\ell|^{-1}\cdot n^{-\delta_\calV(k)}$;
\item Every query for some item in $B$ reveals a value of $0$.
\end{itemize}
Observe that for $\delta_\calV(k)=1/k$, we have actually defined a $k$-well-structured instance in which all agents have the same ordinal preference. As we will see below, this choice of $\delta_\calV$ in fact yields the best lower bound if $\calV$ is the class of unrestricted valuations. However, when $\calV$ is the class of unit-sum valuations, we will have to choose $\delta_\calV$ differently.

Next, we define two types of conditional valuation functions that an agent $i$ may have, depending on the behavior of the algorithm. These functions have to be consistent to the information that is revealed by the queries of the algorithm.
Let $\xi \in (0,1]$ be some constant. \medskip

\noindent 
(T1) If there exists $r \in \{1,\ldots,k+1\}$, such that $i$ is {\em not} queried for any item in $A_r$ and she {\em does not} get an item from $A_r$ either, then $i$'s values are
\begin{itemize}[itemsep=6pt,topsep=6pt,parsep=0pt,partopsep=0pt]
\item at least $\xi \cdot |A_{r-1}|^{-1} \cdot n^{-\delta_\calV(k)}$ for each item in $A_r$ if $r \geq 2$;
\item at least $\xi$ for the item in $A_1$ if $r=1$;
\item $|A_\ell|^{-1}\cdot n^{-\delta_\calV(k)}$ for every item in $A_{\ell}$, for  $\ell\in \{1,\ldots,k+1\}\setminus \{r\}$;
\item $0$ for every item in $B$.
\end{itemize}\smallskip

\noindent (T2) If $i$ is queried for some item in $k$ different sets out of $A_1,\ldots,A_{k+1}$, 
then her values are 
\begin{itemize}[itemsep=6pt,topsep=6pt,parsep=0pt,partopsep=0pt]
\item $|A_\ell|^{-1}\cdot n^{-\delta_\calV(k)}$ for every item in $A_\ell$, for $\ell\in \{1,\ldots,k+1\}$;
\item at most $|A_{k+1}|^{-1}\cdot n^{-\delta_\calV(k)}$ for every item in $B$.
\end{itemize}\smallskip
Observe that the conditions specified in (T1) and (T2) capture all possible cases about the queries of the algorithm and the possible assignments of the items to the agents. 

\begin{theorem}\label{thm:V-lower}
Let $\calV$ be the class of unrestricted or unit-sum valuation functions. 
If there exists a function $\delta_\calV(k) \leq 1/k$ such that it is possible to define valuation functions in $\calV$ of types \emph{(T1)} and \emph{(T2)}, 
the distortion of any matching algorithm which makes $k$ queries per agent is $\Omega\big( \frac{1}{k}\cdot n^{\delta_\calV(k)} \big)$.
\end{theorem}

\begin{proof}
Observe that if the values of the agents for all items in a set $A_\ell$ are consistent to the revealed values, then the total value for all items in set $A_\ell$ is equal to $|A_\ell| \cdot |A_\ell|^{-1}\cdot n^{-\delta_\calV(k)} = n^{-\delta_\calV(k)}$. Since there are valuation functions of types (T1) and (T2), we can indeed define a valuation profile $\vv$ so that the value of every agent for the item she is matched to by the algorithm is exactly the value she would reveal if she was queried for it. So, the items of every set $A_\ell$, $\ell \in \{1,\ldots,k+1\}$, contribute exactly $n^{-\delta_\calV(k)}$ to the social welfare of the matching $Y$ computed by the algorithm, while the items in $B$ contribute $0$. Hence, 
\begin{align*}
\SW(Y|\vv) = (k+1) \cdot n^{-\delta_\calV(k)}. 
\end{align*}
Hence, to show the desired bound on the distortion of the algorithm, it suffices to show that there is always a matching with social welfare $\Omega(1)$.

This is clearly the case when there exists an agent who is not queried for the item in $A_1$ and is not given this item, since her value for it in such a case can be set to be at least $\xi$ using a function of type (T1) for $r=1$. Therefore, since this item can only be given to one agent, in the following we assume that there is {\em at most one} agent who is not queried for it, and if such an agent exists, she must be given the item. 

Consider the set $S_2$ of agents who are not queried for any item in $A_2$ and also do not get any item in $A_2$. 
If $|S_2| \geq |A_2|$, then by defining the valuation function of every agent in $S_2$ to be of type (T1) for $r=2$, we can obtain a matching (in which exactly $|A_2|$ agents in $S_2$ are given a different item from $A_2$) with social welfare at least $|A_2| \cdot \xi \cdot |A_1|^{-1} \cdot n^{-\delta_\calV(k)} = \xi$. The latter follows by the fact that $|A_t| \cdot |A_{t-1}|^{-1} = n^{\delta_\calV(k)}$ for all $t\ge 2$. 
Consequently, it must be $|S_2| < |A_2|$. This further implies that there are at least $n-1 - 2|A_2|$ agents who are queried for some item in each of $A_1$ and $A_2$, and do not get any item in $A_1\cup A_2$. Let $L_2$ be the set of these agents.

Consider the set $S_3 \subseteq L_2$ of agents who are not queried for any item in $A_3$ and also do not get any item in $A_3$. 
Like in the case of $A_3$ and $S_3$ above, if $|S_3| \geq |A_3|$, then by defining the valuation function of every agent in $S_3$ to be of type (T1) for $r=3$, we can obtain a matching (in which $|A_3|$ of the agents in $S_3$ are given a different item from $A_3$) with social welfare at least $|A_3| \cdot \xi \cdot |A_2|^{-1} \cdot n^{-\delta_\calV(k)} = \xi$. So, it must be $|S_3| < |A_3|$, which implies that there are $|L_2|-|A_3|-|S_3| \geq |L_2|-2|A_3|$ agents in $L_2$ who are queried for some item in each of $A_1$, $A_2$ and $A_3$, and do not get any item in $A_1\cup A_2\cup A_3$. Let $L_3$ be the set of these agents.

By induction, this process leads to the existence of the set $L_k$ of agents who have been queried for some item in each of $A_1$, \ldots, $A_k$,  such that $|L_k| \geq n - 1 - 2\sum_{\ell=1}^k |A_\ell|$. Since 
\[
n > 2\sum_{\ell=1}^{k+1} |A_\ell| \,\Rightarrow\, 
n - 1 - 2\sum_{\ell=1}^k |A_\ell| - |A_{k+1}| \geq |A_{k+1}|\,,
\]
there are $|A_{k+1}|$ agents who have not been queried for any item in $A_{k+1}$ and do not get any item in $A_{k+1}$. Thus, by setting their valuation functions to be of type (T1) for $r=k+1$, we can construct a matching with welfare at least $\xi$, completing the proof.
\end{proof}

\cref{thm:V-lower} is actually quite powerful and allows us to prove lower bounds for both unrestricted and unit-sum valuation functions. In particular, it reduces the problem to finding the largest possible $\delta_\calV(k) \leq 1/k$, such that valuation functions in $\calV$ of types (T1) and (T2) can be defined.

\begin{theorem}\label{thm:lower-unconstrained} 
For unconstrained valuation functions, the distortion of any matching algorithm which makes $k$ queries per agent is $\Omega\big(\frac{1}{k} \cdot n^{1/k}\big)$.
\end{theorem}

\begin{proof}
It is straightforward to observe that it is indeed possible to define unconstrained valuation functions of types (T1) and (T2) for the function $\delta_\infty(k) = 1/k$. In particular, the valuation functions are such that
\begin{align*}
u_1(j) =
\begin{cases}
n^{-(r-1)/k}, & \text{ if } j \in A_r \\
n^{-\ell/k}, & \text{ if } j \in A_{\ell}, \ell\neq r \\
0, & \text{ if } j \in B
\end{cases}
\end{align*}
and
\begin{align*}
u_2(j) =
\begin{cases}
n^{-\ell/k}, & \text{ if } j \in A_{\ell}\\
0, & \text{ if } j \in B
\end{cases}
\end{align*}
where $\ell$ is a generic index, while $r$ in the definition of $u_1$ is an index that corresponds to a set $A_r$ such that $i$ is not queried for any item in it and she does not get an item from it either.

Hence, by \cref{thm:V-lower}, any matching algorithm has distortion $\Omega\big(\frac{1}{k} \cdot n^{1/k}\big)$.
\end{proof}

For unit-sum valuations, we have the following bound.

\begin{theorem}\label{thm:lower-unit-sum} 
Let $\xi\in(0,1)$ be a constant. For unit-sum valuation functions, the distortion of any matching algorithm which makes $k \le (1-\xi) n^{1/(k+1)}$ queries per agent is $\Omega\big(\frac{1}{k} \cdot n^{1/(k+1)}\big)$.
\end{theorem}

\begin{proof}
Let $\varepsilon\in(0,1/2)$ be the constant used for defining the sets $A_1, \ldots, A_{k+1}$.
We will show that for $\delta_1(k) = 1/(k+1)$, the following two valuation functions satisfy the unit-sum normalization  and are of types (T1) and (T2), respectively. Then, the statement will follow by \cref{thm:V-lower} by substituting $\delta_1(k)$. 
The functions are defined as 
\begin{align*}
u_1(j) =
\begin{cases}
\frac{1 - k\cdot n^{-1/(k+1)}}{\varepsilon n^{(r-1)/(k+1)}}, & \text{ if } j \in A_r \\
\varepsilon^{-1}n^{-\ell/(k+1)}, & \text{ if } j \in A_\ell, \ell\neq r \\
0, & \text{ if } j \in B
\end{cases}
\end{align*}
and
\begin{align*}
u_2(j) =
\begin{cases}
\varepsilon^{-1}n^{-\ell/(k+1)}, & \text{ if } j \in A_\ell \\
\frac{1-(k+1)\cdot n^{-1/(k+1)}}{|B|}, & \text{ if } j \in B
\end{cases}
\end{align*}
where $\ell$ is a generic index, while $r$ in the definition of $u_1$ is an index that corresponds to a set $A_r$ such that $i$ is not queried for any item in it and she does not get an item from it either.

First, let us verify that both functions satisfy the unit-sum assumption. 
By the definition of $\delta_1(k)$, we have that $|A_\ell| = \varepsilon n^{(\ell-1)/(k+1)}$ for all $\ell \in \{1,\ldots,k+1\}$. Therefore,
\begin{align*}
\sum_{j \in A} u_1(j) 
&= |A_r| \cdot \frac{1 - k\cdot n^{-1/(k+1)}}{\varepsilon n^{(r-1)/(k+1)}} 
+ \sum_{\ell = 1}^{k+1} |A_\ell| \cdot \varepsilon^{-1} n^{-\ell/(k+1)}  - |A_{r}| \cdot\varepsilon^{-1} n^{-{r}/(k+1)}\\
&= 1 - k\cdot n^{-1/(k+1)} + k\cdot n^{-1/(k+1)} = 1,
\end{align*}
and
\begin{align*}
\sum_{j \in A} u_2(j) 
&= \sum_{\ell = 1}^{k+1}  |A_\ell| \cdot \varepsilon^{-1} n^{-\ell/(k+1)} 
+ |B| \cdot \frac{1-(k+1)\cdot n^{1/(k+1)}}{|B|} \\
&= (k+1) \cdot n^{-1/(k+1)} + 1-(k+1)\cdot n^{1/(k+1)} 
= 1.
\end{align*}
Next, we will show that $u_1$ is of type (T1). It suffices to show that the values satisfy the corresponding conditions. We have
\begin{itemize}[itemsep=6pt,topsep=6pt,parsep=0pt,partopsep=0pt]
\item 
For every item $j \in A_r$, if $r \geq 2$:
\begin{align*}
u_1(j) &= \frac{1 - k\cdot n^{-1/(k+1)}}{\varepsilon n^{(r-1)/(k+1)}} 
= \varepsilon^{-1}\big(n^{-(r-1)/(k+1)} - k\cdot n^{-r/(k+1)}\big) \\
&\geq \varepsilon^{-1} \cdot \xi \cdot n^{-(r-1)/(k+1)} 
= \varepsilon^{-1}\cdot \xi \cdot n^{-(r-2)/(k+1)} \cdot n^{-1/(k+1)} 
= \xi \cdot |A_{r-1}|^{-1} \cdot n^{-\delta_1(k)}.
\end{align*}
\item 
For the item $j \in A_1$: $u_1(j) = \varepsilon^{-1}\big(1 - k\cdot n^{-1/(k+1)}\big) \geq \varepsilon^{-1}\xi \ge \xi$.

\item 
For every item $j \in A_\ell$, $\ell \neq r$:  
\begin{align*}
u_1(j) 
&= \varepsilon^{-1}n^{-\ell/(k+1)} 
= \varepsilon^{-1} n^{-(\ell-1)/(k+1)}\cdot n^{-1/(k+1)} 
= |A_\ell|^{-1}\cdot n^{-\delta_1(k)}.
\end{align*}

\item For every item $j \in B$: $u_1(j) = 0$.
\end{itemize}
Finally, we will show that $u_2$ is of type (T2). Similarly to the above case, we have
\begin{itemize}[itemsep=6pt,topsep=6pt,parsep=0pt,partopsep=0pt]
\item 
For every item $j \in A_\ell$:
\begin{align*}
u_2(j) 
&= \varepsilon^{-1} n^{-\ell/(k+1)} 
= \varepsilon^{-1} n^{-(\ell-1)/k} \cdot n^{-1/(k+1)} 
= |A_\ell|^{-1}\cdot n^{-\delta_1(k)}.
\end{align*}

\item 
For every item $j \in B$:
\begin{align*}
u_2(j) &= \frac{1-(k+1)\cdot n^{-1/(k+1)}}{|B|} 
= \frac{1-(k+1)\cdot n^{-1/(k+1)}}{ n - \sum_{\ell=1}^{k+1} n^{(\ell-1)/(k+1)}} \\
&\leq  \frac{1-(k+1)\cdot n^{-1/(k+1)}}{n - (k+1)n^{k/(k+1)}} 
= n^{-1} 
= \varepsilon\cdot\varepsilon^{-1}\cdot n^{-k/(k+1)} \cdot n^{-1/(k+1)} \\
&= \varepsilon\cdot|A_{k+1}|^{-1}\cdot n^{-\delta_1(k)} \le |A_{k+1}|^{-1}\cdot n^{-\delta_1(k)}\,,
\end{align*}
where the first inequality follows by the fact that $n^{x/(k+1)}$ is increasing in $x$. \hfill $\qedhere$
\end{itemize}
\end{proof}

By appropriately setting the value of $k$ in \cref{thm:lower-unconstrained} and \cref{thm:lower-unit-sum}, we  establish that it is impossible to achieve constant distortion without an almost logarithmic number of queries.

\begin{corollary}\label{cor:lower}
Any matching algorithm allowed to make at most
\begin{itemize}[itemsep=6pt,topsep=6pt,parsep=0pt,partopsep=0pt]
\item a constant number $k$ queries per agent has distortion $\Omega(n^{1/k})$ when the valuation functions are unrestricted, and $\Omega(n^{1/(k+1)})$ when they are unit-sum;
 
\item $o \big( \frac{\log{n}}{\log\log{n}} \big)$ queries per agent has  distortion $\omega(\log\log{n})$.
\end{itemize}
\end{corollary}

\section{Two Queries for Unit-sum Valuations}\label{sec:unit-sum}

In this section, we present the \textsc{FirstPositionAdaptive} algorithm (FPA), 
which makes at most two queries per agent and achieves a distortion of $O(n^{2/3}\sqrt{\log{n}})$, when the valuation functions are unit-sum.  
First, we query each agent for their most-preferred item. 
Then, depending on whether the \emph{maximum revealed value} by these queries is at least $n^{-1/3}$, we query the agents for items that are parts of ``large enough'' partial matchings. 
Otherwise, we query everyone at a specific position, and define \emph{simulated values} based on the answers to these queries, ensuring that these values are lower bounds on the corresponding true values. Clearly, the simulated valuation functions are not necessarily unit-sum. 

For the sake of the presentation, we assume that $n$ is a perfect cube, that is, $n=\alpha^3$ for some $\alpha\in \mathbb N$; it is straightforward to extend our analysis to the case where this is not true.

\medskip

{
\begin{mdframed}[backgroundcolor=white!90!gray,
	leftmargin=\dimexpr\leftmargin-20pt\relax,
	rightmargin=\dimexpr\rightmargin+5pt\relax
	skipabove=5pt,skipbelow=5pt]
\textsc{FirstPositionAdaptive} (FPA)\\[-5pt]

\noindent All agents are initially {\em active}.\\[-5pt]

\noindent For every agent  $i$, query $i$ for her top item $j_i^*$; let $v_i^*$ be its value.\\[-5pt]

\noindent If $\max_{i \in N} v_i^* \geq n^{-1/3}$, then:
\begin{itemize}[itemsep=6pt,topsep=6pt,parsep=0pt,partopsep=0pt]
\item For every $\ell \in [n]$, \emph{while} there exists a partial matching $Y_p$ of size 
$|Y_p|\geq n^{1/3}/\sqrt{\log{n}}$ consisting of active agents $i$ matched to items $y_i$ such that agent $i$ ranks item $y_i$ at some position $\ell' \le \ell$, query every $i$ in $Y_p$ for item $y_i$ and make these agents \emph{inactive}.
If necessary, break ties arbitrarily.


\item Output a matching $Y$ that maximizes the social welfare, according to the revealed values due to the above queries.
\end{itemize}

\noindent Else (i.e., $\max_{i \in N} v_i^* < n^{-1/3}$):
\begin{itemize}[itemsep=6pt,topsep=6pt,parsep=0pt,partopsep=0pt]
\item For every agent $i$, query $i$ for the item she ranks at position $n^{1/3}+1$; let $u_i$ be this value.

\item For every agent $i$, define the simulated valuation function $\tilde{v}_i$:
\begin{itemize}[itemsep=4pt,topsep=4pt,parsep=0pt,partopsep=0pt]
\item $\tilde{v}_i(j_i{^*}) = v_i^{*}$;
\item $\tilde{v}_i(j) = u_i$ for every item $j$ that $i$ ranks at position $\ell \in \{2, \ldots, n^{1/3}+1\}$;
\item $\tilde{v}_i(j) = 0$ for every item $j$ that $i$ ranks at position $\ell \in \{n^{1/3}+2, \ldots, n\}$.
\end{itemize}

\item For every agent $i$ such that $u_i < \frac{1}{2} n^{-1}$, modify $\tilde{v}_i$ so that:
\begin{itemize}[itemsep=4pt,topsep=4pt,parsep=0pt,partopsep=0pt]
\item $\tilde{v}_i(j) = \frac{1}{3} n^{-1/3}$ for every item $j$ that $i$ ranks at position $\ell \in \{2, \ldots, \frac{1}{4} n^{1/3} \}$.
\end{itemize}
\item Output a matching $Y \in \arg\max_{Z} \SW(Z | \tilde{\vv})$.
\end{itemize}
\end{mdframed}
}

\bigskip

The distortion of the algorithm is given by the following theorem.

\begin{theorem}\label{thm:dist_fpa}
For unit-sum valuation functions, the distortion of FPA is $O(n^{2/3}\sqrt{\log{n}})$.
\end{theorem}

\begin{proof}
Let $\vv$ be a valuation profile. Denote by $Y$ the output of the algorithm when given as input the ordinal profile $\succ_\vv$, and by $X=(x_i)_{i \in N}$ an optimal matching for $\vv$. We consider two main cases, depending on the value $\max_{i \in N} v_{i}^*$ that the algorithm learns with the first query. 

\bigskip

\noindent \underline{Case 1: $\max_{i \in N} v_{i}^* \geq n^{-1/3}$.}\ \ \ 
The algorithm makes a second query to an agent $i$ for some item $j$ only if the pair $(i,j)$ is part of a partial matching of size at least $n^{1/3}$, involving only active agents, i.e., agents who have not been included in such a partial matching in any previous step. 
Let $Y_1, \ldots, Y_\lambda$ be all the partial matchings considered throughout the execution of the algorithm. 
By definition, each such partial matching contains at least $n^{1/3}/\sqrt{\log{n}}$ agents and  an agent is contained in at most one of these matchings. Thus, it holds that $\lambda < n^{2/3}\sqrt{\log{n}}$. 

We partition the agents into two sets.
The set $H$ of agents $i$ that are queried only for items they rank \emph{at least as high} as the item $x_i$ they receive in the optimal matching $X$. Some agents in $H$ are possibly queried twice for their best item.
The set $L$ of agents $i$ that are queried for an item they rank \emph{lower} than $x_i$ or are queried only \emph{once}.
We can write the social welfare of $X$ as
\begin{align*}
\SW(X|\vv) = \sum_{i \in H} v_i(x_i) + \sum_{i \in L} v_i(x_i) \,.
\end{align*}
We will bound each term on the right-hand side separately. 
For the first term, we have:
\begin{align*}
\sum_{i \in H} v_i(x_i) &\leq \sum_{i \in H} v_i(y_i) \leq \sum_{t=1}^\lambda \sum_{i \in Y_t} v_i(y_i) 
\le \lambda \max_t \!\sum_{i \in Y_t} v_i(y_i) < n^{2/3} \sqrt{\log{n}} \cdot \SW(Y|\vv)\,.
\end{align*}
The first inequality holds because $y_i \succcurlyeq_i x_i$ for every  $i \in H$. 
The second inequality holds because the agents in $H$ are queried only if they are included in one of the partial matchings $Y_1, \ldots, Y_\lambda$. 
The last inequality follows from the bound on $\lambda$ established above, and the fact that $\max_t \sum_{i \in Y_t} v_i(y_i)$ is trivially upper bounded by the social welfare of $Y$.

To bound the second term, let $Y^{(\ell)}$ be the restriction of $Y$ containing only the agents $i \in L$ for whom $x_i$ is at position $\ell$. 
It holds that $|Y^{(\ell)}| < n^{1/3}/\sqrt{\log{n}}$, or else the algorithm would have queried the agents in $Y^{(\ell)}$ for their optimal items, contradicting their membership in $L$. So, we get that
\begin{align*}
\sum_{i \in L} v_i(x_i) &= \sum_{\ell =1}^{n} \sum_{i \in Y^{(\ell)}} v_i(x_i) < \sum_{\ell =1}^{n} n^{1/3} \sqrt{\log{n}} \, \frac{1}{\ell} 
< \frac{n^{1/3}}{\sqrt{\log{n}}} 2\log n = 2n^{1/3}\sqrt{\log{n}} \,,
\end{align*}
where the first inequality follows from the unit-sum normalization; in particular, any agent's value for an item at position $\ell$ is at most $1/\ell$. The second inequality is a simple bound on the harmonic numbers: 
$\sum_{i =1}^{n} i^{-1}  < 2\log_2 n$, for $n\ge 2$.

Further, since $\max_{i \in N} v_{i}^* \geq n^{-1/3}$, we have that $\SW(Y|\vv) \geq n^{-1/3}$. Thus,
\begin{align*}
\sum_{i \in L} v_i(x_i) \leq 2 n^{2/3}\sqrt{\log{n}}\cdot \SW(Y|\vv) \,.
\end{align*}
Putting everything together, the distortion of the algorithm in this case is at most $2 n^{2/3}\sqrt{\log{n}}$.

\bigskip

\noindent \underline{Case 2: $\max_{i \in N} v_{i}^* < n^{-1/3}$.}\ \ \ 
We partition the set of agents into two sets, depending on whether their value for the item they rank at position $n^{1/3}+1$ is at most $\frac{1}{2}n^{-1}$. In particular, let $R = \{ i\in N: u_i < \frac{1}{2}n^{-1}\}$. We can write the optimal social welfare of $X$ as
\begin{align*}
\SW(X | \vv) = \sum_{i \in R} v_i(x_i) + \sum_{i \in N\setminus R} \!v_i(x_i) \,.
\end{align*}
We will bound each term separately. 
For the first term, since $\max_{i \in N} v_{i}^* < n^{-1/3}$, we clearly have that
\begin{align*}
\sum_{i \in R} v_i(x_i) \leq  \max_{i \in N} v_{i}^*|R| \leq   n^{-1/3}|R| \,.
\end{align*}
Consider an arbitrary agent $i \in R$ and denote by $j_{i,\ell}$ the item she ranks at position $\ell$; hence, $j_i^* = j_{i,1}$. 
We will first show that $v_i\big(j_{i,\frac{1}{4}n^{1/3}}\big) \geq \frac{1}{3}\cdot n^{-1/3} = \tilde{v}_i\big(j_{i,\frac{1}{4}n^{1/3}}\big)$.
Since $u_i = v_i\big(j_{i,n^{1/3}+1}\big) < \frac{1}{2} n^{-1}$, we have that
\begin{align*}
\sum_{\ell = n^{1/3}+1}^n v_i(j_{i,\ell}) \leq (n-n^{1/3}-1)  u_i < \frac{1}{2} \,,
\end{align*}
and thus, by the unit-sum normalization assumption, we also have that
\begin{align*}
\sum_{\ell = 1}^{n^{1/3}} v_i(j_{i,\ell}) \geq \frac{1}{2} \,.
\end{align*}
Since $v_i(j_{i,\ell}) \leq v_i(j_{i,1}) < n^{-1/3}$ for every $\ell \in \big\{1, \ldots, \frac{1}{4}n^{1/3}-1\big\}$ and $v_i\big(j_{i,\frac{1}{4}n^{1/3}}\big) \geq v_i(j_{i,\ell})$ for every $\ell \in \big\{ \frac{1}{4}n^{1/3}, \allowbreak \ldots, n^{1/3} \big\}$, we obtain
\begin{align*}
v_i\big(j_{i,\frac{1}{4}n^{1/3}}\big) &\geq \frac{\frac{1}{2} - \big(\frac{1}{4}n^{1/3}-1\big) n^{-1/3} }{ \frac{3}{4}n^{1/3} } \geq \frac{1}{3} n^{-1/3} = \tilde{v}_i\big(j_{i,\frac{1}{4}n^{1/3}}\big) \,.
\end{align*}
where the second inequality is a matter of simple calculations.
So, all the agents in $R$ have value at least $\frac{1}{3} n^{-1/3}$ for the items they rank at positions up to $\frac{1}{4}n^{1/3}$. This implies that the simulated valuation functions, defined by the algorithm, are lower bounds to the real valuation functions.

By Hall's Theorem \cite{hall1935representatives}, it is easy to see that there exists a matching of size $\min\big\{|R|,\frac{1}{4}n^{1/3}\big\}$ where each agent in $R$ is matched to an  item she  ranks at the first $\frac{1}{4}n^{1/3}$ positions. Moreover, $Y$ maximizes the social welfare according to the simulated valuation functions. Hence,
\begin{align*}
\SW(Y | \vv) \geq \SW(Y | \tilde{\vv}) \geq \frac{1}{3}n^{-1/3} \min\Big\{|R|,\frac{1}{4}n^{1/3}\Big\} \,.
\end{align*}
If $|R| <  \frac{1}{4}n^{1/3}$, then $\SW(Y | \vv) \geq \frac{1}{3} |R| n^{-1/3}$, and thus
\begin{align*}
\sum_{i \in R} v_i(x_i) \leq 3 \cdot \SW(Y | \vv) \,.
\end{align*}
Otherwise, $\SW(Y | \vv) \geq 1/12$, and since $|R| \leq n$, we obtain 
\begin{align*}
\sum_{i \in R} v_i(x_i) \leq 12 n^{2/3} \cdot \SW(Y | \vv) \,.
\end{align*}

For the second term, we further partition  $N\setminus R$ into two sets depending on the position of the $x_i$s. 
In particular, $H$ is the set of agents $i \in N\setminus R$ who rank $x_i$ at some position $\ell \leq n^{1/3}$, and 
$L$ is the set  of remaining agents $i \in (N\setminus R) \setminus H$ (who rank $x_i$ at some position $\ell > n^{1/3}$).
Hence,
\begin{align*}
\sum_{i \in N\setminus R} v_i(x_i) = \sum_{i \in H} v_i(x_i) + \sum_{i \in L} v_i(x_i) \,.
\end{align*}
First consider the agents in $H$. Since $\max_{i \in N} v_{i}^* < n^{-1/3}$, 
\begin{align*}
\sum_{i \in H} v_i(x_i) \leq   \max_{i \in N} v_{i}^* |H| <   n^{-1/3} |H| \,.
\end{align*}
Consider any agent $i \in H$ and any item $j$ that $i$ ranks at some position $\ell \leq n^{1/3}$.  Since $u_i$ is the value of $i$ for the item she ranks at position $n^{1/3}+1$, we clearly have that $v_i(j) \geq \allowbreak u_i = \tilde{v}_i(j) \geq \frac{1}{2}n^{-1}$. 
Note that there exists a partial matching of size $|H|$ according to which {\em all} agents of $H$ are matched to items they rank at the first $n^{1/3}$ positions; e.g., the restriction of $X$ on $H$. 
Since $Y$ maximizes the social welfare for the simulated valuation functions, we get
\begin{align*}
\SW(Y|\vv) \geq \SW(Y | \tilde{\vv}) \geq \frac{1}{2} n^{-1} |H| \,,
\end{align*}
which immediately implies that
\begin{align*}
\sum_{i \in H} v_i(x_i) \leq 2 n^{2/3} \cdot \SW(Y|\vv) \,.
\end{align*}
Finally, consider the agents in $L$, and distinguish the following two cases depending on the size of $L$.
\begin{itemize}[leftmargin=15pt,itemsep=6pt,topsep=6pt,parsep=0pt,partopsep=0pt]
\item $|L| \leq n^{1/3}$. 
Since there are at least $n^{1/3}$ different items within the first $n^{1/3}$ positions of each agent in $L$, by Hall's Theorem, there exists a matching $Y'$ according to which {\em all} agents in $L$ receive such an item, i.e., every $i \in L$ has (simulated) value at least $u_i$ for the item she gets in $Y'$. Combining this with the optimality of $Y$ for the simulated valuation functions and the fact that the latter lower bound the real valuation functions, we have
\begin{align*}
\SW(Y|\vv) &\ge \SW(Y|\tilde{\vv}) \ge \SW(Y'|\tilde{\vv}) \ge \sum_{i \in L} u_i \ge \sum_{i \in L} v_i(x_i) \,,
\end{align*}
where the last inequality follows by the definition of $L$.

\item $|L| > n^{1/3}$. Denote by $S_L$ the $|S_L| = n^{1/3}$ agents with the highest values $u_i$ among all the agents in $L$. We may repeat the above argument for $S_L$ instead of $L$ to get $\SW(Y|\vv)  \ge \sum_{i \in S_L} u_i$. Then,
\begin{align*}
\SW(Y|\vv) \geq  n^{1/3}  \min_{i \in S_L} u_i \geq n^{1/3}  \max_{i \in L \setminus S_L} \! u_i \,.
\end{align*}
On the other hand, we have
\begin{align*}
\sum_{i \in L} v_i(x_i) 
&\leq \sum_{i \in S_L} u_i + (|L|-|S_L|)  \max_{i \in L \setminus S_L}\! u_i \\
&\leq \sum_{i \in S_L} u_i + n  \max_{i \in L \setminus S_L} \! u_i \\
&\leq (1 + n^{2/3}) \cdot \SW(Y|\vv).
\end{align*}
\end{itemize}
Therefore, the distortion of the algorithm is at most $16n^{2/3} + 1$ in case 2. Together with case 1, we obtain the desired bound of $O(n^{2/3}\sqrt{\log{n}})$.
\end{proof}

\section{A General Framework for $\lambda$-TSF}\label{sec:extensions}

In this section we generalize $\lambda$-TSF, our algorithm from \cref{sec:upper}, to work for a much broader class of problems, where we are given the ordinal preferences of the agents and access via queries to their cardinal values. We begin with the following general \emph{full information} problem of maximizing an additive objective over a family of combinatorial structures defined on a weighted graph: 

\medskip

\noindent
\textbf{Max-on-Graphs}: Given a (directed or undirected) weighted graph $G = (U, E, w)$ and a concise description of the set $\mathcal{F}\subseteq 2^E$ of feasible solutions, find a solution $S \in \allowbreak \arg\max_{T\in\mathcal{F}} \sum_{e\in T} w(e)$.

\medskip
\noindent
Note that one-sided matching is a special case; $G$ is the complete bipartite graph on $N$ and $A$, the weight of an edge $\{i,j\}$ is $v_i(j)$, and $\mathcal{F}$ contains  the perfect matchings of $G$.

What we are really interested in is the social choice analog of Max-on-Graphs where the weights (defined in terms of the  valuation functions of the agents) are not given! Instead, we know the ordinal preferences of each agent/node for other nodes (corresponding to items or other agents).

\medskip

\noindent\textbf{Ordinal-Max-on-Graphs}:
Here $U = N \cup A$, where $N$ is the set of agents and $A$ is the (possibly empty) set of items; when $A\neq \emptyset$, we assume that $G$ is a bipartite graph with independent sets $N, A$.
Although $G = (U, E)$ is given without the weights, it is assumed that for every $i\in N$ there exists a (private) valuation function $v_i: U \rightarrow \RR_{\geq 0}$, so that 
\vspace{-2pt}\begin{align}\label{eq:weights}
w(e) =
\begin{cases}
v_i(j), & \text{ if } i\in N, j\in A \text{ and } e = \{i,j\}\hspace{12pt} \text{{\footnotesize (Bipartite agent--item case)}} \\
v_i(j) + v_j(i), & \text{ if } i, j\in N \text{ and } e = \{i,j\} \qquad\quad \text{{\footnotesize (Undirected case)}} \\
v_i(j), & \text{ if } i, j\in N \text{ and } e = (i,j) \qquad\quad \text{{\footnotesize\ (Directed case)}} \,.
\end{cases}
\end{align}

\noindent
We are also given the {\em ordinal profile} $\succ_\vv = (\succ_i)_{i \in N}$ induced by $\vv = (v_i)_{i \in N}$ and a concise description of the set $\mathcal{F}\subseteq 2^E$ of feasible solutions. The goal is again to find  $S \in \arg\max_{T\in\mathcal{F}} \sum_{e\in T} w(e)$. 

\medskip

Notice that for Ordinal-Max-on-Graphs to make sense, $\mathcal{F}$ should be independent of $w$. For example, if only sets of weight exactly $B$ are feasible, then it is impossible to find even one feasible set without the exact cardinal information in our disposal. Still, it is  clear that the above algorithmic problem is very general and captures a huge number of maximization problems on graphs. Of course, not all such problems have a natural interpretation where the vertices are agents with preferences. Before we state the main result of this section, we give several examples that have been studied in the computational social choice literature. 

\medskip

\noindent\emph{\textbf{General Graph Matching}}: Given an undirected weighted graph $G = (U, E, w)$, find a matching of maximum weight, i.e., $\mathcal{F}$ contains the matchings of $G$.  
In its social choice analog, $U=N$ and $w(\cdot)$ is defined according to the second branch of \eqref{eq:weights}.
A special case of this problem in which $G = (U_1\cup U_2, E, w)$ is a bipartite graph, is the celebrated \textbf{\emph{two-sided matching}} problem \citep{GaleShapley62,RothSotomayor92}. 

\medskip

\noindent\emph{\textbf{Two-sided Perfect Matching}}: A variant of two-sided matching, where $|U_1|=|U_2|$ and only perfect matchings are feasible. \medskip

\noindent\emph{\textbf{Max $k$-Sum Clustering}}: Given an undirected weighted graph $G = (U, E, w)$, where $|U|$ is a multiple of $k$, partition $U$ into $k$ equal-sized clusters in order to maximize the weight of the edges inside the clusters. That is, $\mathcal{F}$ contains, for each partition of $U$ into $k$ equal-sized clusters, the set of edges that do not cross clusters. This problem generalizes two-sided perfect matching; see \citep{anshelevich2016blind}. In its social choice analog, $U=N$ and the weights are defined according to the second branch of \eqref{eq:weights}.

\medskip

\noindent\emph{\textbf{General Resource Allocation}}: Given a bipartite weighted graph $G = (U_1\cup U_2, E, w)$, assign each node of $U_2$ to (only) one neighboring node in $U_1$ so that the total value of the corresponding edges is maximized. There may be additional combinatorial constraints on this assignment, e.g., no more than $\beta_i$ nodes of $U_2$ may be assigned to node $i\in U_1$. That is, $\mathcal{F}$ contains the sets of edges that define the partitions of $U_2$ into $|U_1|$ parts that also satisfy the additional constraints. This problem generalizes one-sided matching. In its social choice analog, $U_1=N$, $U_2=A$ and $w(\cdot)$ is defined according to the first branch of \eqref{eq:weights}.

\medskip

\noindent\emph{\textbf{Clearing Kidney $\ell$-Exchanges}}: Given a directed weighted graph $G = (U, E, w)$, find a collection of vertex-disjoint cycles of length at most $\ell$ so that their total weight is maximized;
see \citep{abraham2007clearing}. 
Here, $\mathcal{F}$ contains the edge set of any such collection of short cycles. 
In its social choice analog, $U=N$ and $w(\cdot)$ is defined according to the third branch of \eqref{eq:weights}.

\medskip

We use a variant of $\lambda$-TSF, $(\lambda,\mathcal{A})$-TSF, that takes as an additional input an approximation algorithm $\mathcal{A}$ for the problem at hand. 
There are two main differences from $\lambda$-TSF. The simpler one is about the last step; instead of computing a maximum matching, $\mathcal{A}$ is used to compute an (approximately) optimal solution  with respect to the simulated valuation functions. The other difference is more subtle. Now we do not want to ask each agent $i$ for her top  element of $U$, but rather for her top element that induces an edge included in some feasible solution. It is not always trivial to find this element for each agent, but often it can be done in polynomial time; see \cref{cor:binary-general1} for such examples. 

\medskip

\begin{mdframed}[backgroundcolor=white!90!gray,
	leftmargin=\dimexpr\leftmargin-20pt\relax,
	rightmargin=\dimexpr\rightmargin+5pt\relax
	skipabove=5pt,skipbelow=5pt]
\hspace{-3pt}$(\lambda,\mathcal{A})$-TSF \\[-5pt]

\noindent Let $\alpha_\ell = r^{-\ell/(\lambda+1)}$ for $\ell \in \{0, \ldots, \lambda\}$, where $r = \max_{T\in\mathcal{F}} |T|$.\\[-5pt]

\noindent For every  $i \in N$:
\begin{itemize}[itemsep=6pt,topsep=6pt,parsep=0pt,partopsep=0pt]
\item Find $i$'s highest-ranked element in $U$, $j_i^*$, that defines an edge contained in \emph{some} feasible solution. That is, $j_i^*$ is such that there exist $T\in \mathcal{F}$ such that $\{i, j_i^*\}\in T$ (or $(i, j_i^*)\in T$ in the directed case) and any other element with this property is less preferred for $i$.
\item Learn the value $v_i^*$ of $i$ for $j_i^*$, and let $Q_{i,0} = \{j_i^*\}$, $\tilde{v}_i(j_i^*) = \alpha_0 \cdot v_i^* = v_i^*$. (We assume that when $A\neq \emptyset$, $v_i(j) = 0$ for all $j\in N$).
\item For every $\ell \in \{1, \ldots, \lambda\}$, using binary search, compute the set 
\begin{align*}
Q_{i,\ell} = \left\{j \in A: v_{i}(j) \in \left[ \alpha_\ell \cdot v_i^*, \alpha_{\ell-1} \cdot v_i^* \right) \right\}
\end{align*}
and let $\tilde{v}_{i}(j) = \alpha_\ell \cdot v_i^*$ for every $j \in Q_{i,\ell}$.
\item Let $Q_i = \bigcup_{\ell = 0} ^{\lambda} Q_{i,\ell}$ and set $\tilde{v}_{i}(j) = 0$  for every item $j \in A \setminus Q_i$.
\end{itemize}

\noindent Using \eqref{eq:weights} and the simulated valuation functions $\tilde{v_i}, i\in N$, compute the \emph{simulated edge weights} $\tilde{w}(e), e\in E$. \\[-5pt]

\noindent Return $\mathcal{A}(\tilde{G})$, where $\tilde{G} = (U, E, \tilde{w})$. 
\end{mdframed}

\smallskip

Note that the step of finding the $j_i^*$s is not given explicitly as it has to be adjusted for the particular problem at hand. As a concrete non-trivial example, consider the perfect matching variant of general graph matching, where we only care about \emph{perfect} matchings. In this case, we can check whether an edge $\{i,j\}$ belongs to a perfect matching by removing both $i$ and $j$ and then running the blossom algorithm of \citet{edmonds_1965} on the remaining graph (with all weights set to $1$). So, by repeatedly using this subroutine for an agent $i$ starting from her top element and going down her preference list, we can find $j_i^*$ in polynomial time and then make a query for it.

For the following theorem, we assume that the optimization problem $\Pi$ is a special case of Max-on-Graphs with $\max_{T\in\mathcal{F}} |T| = r$. The parameter $r$ allows for a more refined statement; while for the general Max-on-Graphs $r$ may be $\Theta(|U|^2)$, in most cases it is only $O(|U|)$.  We further assume that we can efficiently check whether an edge $e$ belongs to a feasible solution. If not, the distortion guarantee of the theorem is still true, but there is no guarantee about the running time of $(\lambda,\mathcal{A})$-TSF.

\begin{theorem}\label{thm:binary-general}
Suppose $\Pi$ is as described above. If $\mathcal{A}$ is a (polynomial-time) $\rho$-approximation algorithm for $\Pi$ in the full information setting, then 
$(\lambda,\mathcal{A})$-TSF asks $1 + \lambda + \lambda\log{r}$ queries and achieves distortion at most $3 \rho \, r^{\frac{1}{\lambda+1}}$ for the social choice analog of $\,\Pi$ (in polynomial time).
\end{theorem}

\begin{proof}
Let $X \subseteq E$ be an optimal solution according to the valuation functions $v_i$, and $Z$ be the solution returned by $(\lambda,\mathcal{A})$-TSF. Also, let $Y$ be an  optimal solution with respect to the simulated valuation functions.

In order to unify the notation for the three definitions of edge weights in (1), we write $e = \langle i, j\rangle$ to mean
\begin{itemize}[itemsep=6pt,topsep=6pt,parsep=0pt,partopsep=0pt]
\item[(i)] $e = \{i, j\}$ with  $i\in N, j\in A$ when $A\neq\emptyset$;
\item[(ii)] $e = \{i, j\}$ when $A=\emptyset$ and $G$ is undirected; 
\item[(iii)] $e = (i, j)$ when $A=\emptyset$ and $G$ is directed. 
\end{itemize}
Using this notation we can define $X_i  = \{j\in U\,:\, \langle i, j\rangle\in X\}$, for $i\in N$, and write the optimum as
\begin{align*}
\sum_{e\in X} w(e) &= \sum_{\langle i, j\rangle\in X} \!v_i(j) 
= \sum_{i \in N} \sum_{j\in X_i\setminus Q_i} \!v_i(j) + \sum_{i \in N} \sum_{j\in X_i\cap Q_i} \!v_i(j) \,.
\end{align*}
We will bound the two terms separately. We begin with the first one:
\begin{align*}
\sum_{i \in N} \sum_{j\in X_i\setminus Q_i} \!v_i(j)
&< {\alpha_{\lambda}} \sum_{i \in N} \sum_{j\in X_i\setminus Q_i} \!v_i^*  
\le {\alpha_{\lambda}} \!\sum_{\langle i, j\rangle\in X} \!v_i^* 
\leq 2 r {\alpha_{\lambda}}  \max_{k \in N} v_k^* 
\leq 2 r {\alpha_{\lambda}} \rho \sum_{e\in Z} w(e) \,.
\end{align*}
The first inequality follows directly by the definition of $Q_i$. The second inequality follows by extending the scope of the summation to include all (possibly unordered) pairs in $X$. For the third inequality, it suffices to notice that we simultaneously upper bound the number of terms in the summation by $2 |X| \le 2r$ and each $v_i^*$ by their maximum. 
Finally, the last inequality follows by the fact that the optimal value with the simulated valuation functions is at least $\max_{k \in N} v_k^*$ and, thus, the solution $Z$ returned by the algorithm achieves at least a $\rho$-approximation of that.

For the second term we have
\begin{align*}
\sum_{i \in N} \sum_{j \in X_i\cap Q_i} \!v_{i}(j)
= \sum_{i \in N} \sum_{\ell = 0}^k \sum_{j \in X_i \cap Q_{i,\ell}} \!\!v_{i}(j) \,.
\end{align*}
Now, let us assume that $\lambda>0$; we will deal with the simpler case where $\lambda=0$ later.
By definition, for any $\ell \in \{1, \ldots, \lambda\}$ and any $j \in Q_{i,\ell}$, we have that $v_{i}(j) \leq \alpha_{\ell-1} \cdot v_i^* = \frac{\alpha_{\ell-1}}{\alpha_{\ell}} \cdot \alpha_\ell \cdot v_i^* =  \tilde{v}_{i}(j) / \alpha_1$.  Also, for $Q_{i,0} = \{j_i^*\}$, we have $v_{i}(j_i^*) = \tilde{v}_{i}(j_i^*) \leq \tilde{v}_{i}(j_i^*) / \alpha_1$. Hence,

\begin{align}
\sum_{i \in N} \sum_{j \in X_i\cap Q_i} \!v_{i}(j)
&\leq \alpha_1^{-1} \sum_{i \in N} \sum_{j \in X_i \cap Q_i} \!\tilde{v}_{i}(j)
= \alpha_1^{-1} \sum_{i \in N} \sum_{j \in X_i \cap Q_i} \!\tilde{v}_i(j)  + \alpha_1^{-1} \sum_{i \in N} \sum_{j \in X_i \setminus Q_i} \!0 \nonumber \\
&= \alpha_1^{-1} \!\!\sum_{\langle i, j\rangle\in X} \!\tilde{v}_i(j) = \alpha_1^{-1} \sum_{e\in X} \tilde{w}(e)
\le \alpha_1^{-1} \sum_{e\in Y} \tilde{w}(e) \le  \alpha_1^{-1} \rho \sum_{e\in Z} \tilde{w}(e)\nonumber \\
 &\le  \alpha_1^{-1} \rho \sum_{e\in Z} {w}(e) \,. \label{eq:alpha1}
\end{align}
The second inequality follows from the optimality of $Y$ with respect to the simulated valuation functions. The third inequality follows directly from the approximation guarantee of $\mathcal{A}$: $Z$ attains a $\rho$ approximation of the value achieved by $Y$. Finally, the last inequality follows from the fact that $\tilde{v}_{i}(j) \leq v_{i}(j)$ for every $i, j$, and thus $\tilde{w}(e) \leq w(e)$ for all $e\in E$.

Now we can put everything together:
\begin{align}
\sum_{e\in X} w(e) &\leq \left(2 r {\alpha_{\lambda}} \rho  +  \alpha_1^{-1} \rho\right) \sum_{e\in Z} w(e) 
 = 3\rho r^{\frac{1}{\lambda+1}} \sum_{e\in Z} w(e) \,, \label{eq:rho_dist}
\end{align}
and this settles the bound on the distortion of $(\lambda,\mathcal{A})$-TSF when $\lambda>0$. 

When $\lambda=0$, we can repeat the derivation of \eqref{eq:alpha1} but without the factor of $\alpha_1^{-1}$, as this is only needed for the simulated value of items in $Q_{i,\ell}$ for $\ell>0$ and now these sets are empty. So, $\sum_{i \in N} \sum_{j \in X_i\cap Q_i} \!v_{i}(j) \le  \rho \sum_{e\in Z} {w}(e)$ and then the analog of \eqref{eq:rho_dist} is 
\begin{align*}
\sum_{e\in X} w(e) \leq \left(2 r {\alpha_{0}} \rho  +   \rho\right) \sum_{e\in Z} w(e) \le 3\rho r^{\frac{1}{0+1}} \sum_{e\in Z} w(e) \,.
\end{align*}

About the running time, it is easy to see that all steps, except running $\mathcal{A}(\tilde{G})$ and finding $j_i^*$ for each $i\in N$, can be done in polynomial time (in particular $O(|U|\log^2|U|)$). Let $t(|U|)$ be the time needed to  check whether an edge can be extended to a feasible solution. Then, finding all the $j_i^*$s can be done in  time$O(|U|^2 t(|U|))$, as we need to perform the check for at most $|U|-1$ elements per $i\in N$. Hence, if both the feasibility check and $\mathcal{A}$ run in polynomial time, then $(\lambda,\mathcal{A})$-TSF runs in polynomial time as well.
\end{proof}

For the problems defined above, we can get the following. 

\begin{corollary}\label{cor:binary-general1}
By choosing $\mathcal{A}$ appropriately, $(\lambda,\mathcal{A})$-TSF  asks at most $1 + \lambda + \lambda\log{|U|}$ queries and achieves distortion at most
\begin{itemize}[leftmargin=15pt,itemsep=6pt,topsep=6pt,parsep=0pt,partopsep=0pt]
\item  $3\big( \frac{|U|}{2}\big)^{\frac{1}{\lambda+1}}$ in polynomial time for
one-sided matching (thus qualitatively retrieving \cref{thm:lambdaTSFdistortion}), two-sided matching, general graph matching, and two-sided perfect matching,
\item $3 |U_2|^{\frac{1}{\lambda+1}}$ for general resource allocation;
\item $3 |U|^{\frac{2}{\lambda+1}}$ for max $k$-sum clustering;
\item $(4.5 + \varepsilon) |U|^{\frac{1}{\lambda+1}}$  for clearing kidney $3$-exchanges and $(7 + \varepsilon) |U|^{\frac{1}{\lambda+1}}$ for clearing kidney $4$-exchanges in polynomial time, for any constant $\varepsilon\in(0,1)$.
\end{itemize}
\end{corollary}


\begin{proof}
We begin with one-sided matching, two-sided matching, and general graph matching. First notice that the size of any matching is at most $|U|/2$ and thus $r\le |U|/2$. Then, by using an exact algorithm  $\mathcal{A}$ for computing maximum weight matchings, such as the blossom algorithm \citep{edmonds_1965}, \cref{thm:binary-general} directly implies distortion at most $3(|U|/2)^{\frac{1}{\lambda+1}}$.
Regarding the running time, observe that any edge is already a feasible solution, and thus $j_i^*$ is indeed $i$'s most preferred alternative. Since the blossom algorithm runs in polynomial time, we get that $(\lambda,\mathcal{A})$-TSF runs in polynomial time as well.

For two-sided perfect matching the argument is as above but one needs to argue about efficiently checking whether a given  edge extends to a feasible solution. This, however, is already discussed right after the description of  $(\lambda,\mathcal{A})$-TSF.

For general resource allocation, we only need to see that an assignment is fully determined by exactly $|U_2|$ edges matching the items to the agents. That is, $r = |U_2|$. Since we do not deal with the running time in this case, we may assume an algorithm $\calA$ that solves the full-information problem optimally. Then, \cref{thm:binary-general} implies distortion at most  $3 |U_2|^{\frac{1}{\lambda+1}}$. It should be noted here that, depending on the additional constraints imposed by $\mathcal{F}$, the computation of an assignment may vary from easy (e.g., no constraints) to strongly NP-hard (e.g., the items assigned to each agent should form an independent set in a given graph $H$ on $U_2$). 

For max $k$-sum clustering, again we do not deal with the running time. Thus, it suffices to use $|U|^2$ as a straightforward upper bound for $r$ and the distortion bound follows. 

Finally, for clearing kidney $\ell$-exchanges, $\ell\in\{3,4\}$, notice that the number of edges defining a collection of disjoint cycles can be at most $|U|$ and thus $r\le |U|$. Fix a constant $\varepsilon > 0$. 
For clearing kidney $3$-exchanges (resp.~$4$-exchanges),
as $\calA$  we use the polynomial-time $(1.5+\delta)$-approximation (resp.~$(7/3+\delta)$-approximation) algorithm of \citet{JiaTWZ17} with $\delta = \varepsilon/3$. Thus, \cref{thm:binary-general} implies distortion at most  
\[3(1.5+\varepsilon/3)|U|^{\frac{1}{\lambda+1}} = (4.5+\varepsilon)|U|^{\frac{1}{\lambda+1}}\,,\]
for $3$-exchanges and 
\[3(7/3+\varepsilon/3)|U|^{\frac{1}{\lambda+1}} = (7+\varepsilon)|U|^{\frac{1}{\lambda+1}}\,,\]
for $4$-exchanges. Note that the problem is NP-hard, even when $k=3$ \citep{abraham2007clearing}. 
Regarding the running time, we can efficiently check whether an edge $(i,j)$ is in a feasible solution, as it is equivalent to checking whether $(i,j)$ belongs to a cycle of length at most $\ell$. For instance, we may find a shortest $(j,i)$-path in the unweighted version of $G$; $(i,j)$ is in a feasible solution if and only if the length of this shortest path is at most $\ell - 1$.
Since $\calA$ also runs in polynomial time, we get that $(\lambda,\mathcal{A})$-TSF runs in polynomial time as well.
\end{proof}

\section{Conclusion and Open Problems}\label{sec:open}

Our work is the first to study the interplay between elicited information and distortion in one-sided matching, as well as more general graph problems. We have shown several tradeoffs, both in term of possible distortion guarantees, and inapproximability bounds. Our results suggest that using only a small number of queries per agent can lead to significant improvements on the distortion. 

As future directions, first it would be very interesting to see if we can come up with algorithms that match the lower bounds of \cref{thm:lower-unconstrained}. We managed to do that for the class of $k$-well-structured instances, but whether it is possible to achieve that for any instance remains to be seen. Perhaps a slightly less ambitious open problem would be to design an algorithm that outperforms the two-queries algorithm presented in \cref{sec:unit-sum} in terms of the achievable tradeoffs, for agents with unit-sum valuation functions. Another interesting avenue would be to consider randomized algorithms, either in the selection of the matching, or the process of querying the agents, and see if we can obtain significant improvements. Finally, going beyond one-sided matching, one could study more general programs, such as those discussed in \cref{sec:extensions} and design tailor-made algorithms with improved tradeoffs between distortion and number of queries per agent.


\bibliographystyle{plainnat}
\bibliography{references}

\end{document}